\documentclass[10pt,reqno]{article}
\usepackage{fullpage} 
\usepackage{amssymb} 
\usepackage{graphicx} 
\usepackage{amsmath} 
\usepackage{amsthm,amssymb,latexsym} 
\usepackage{amstext, amsfonts,amsbsy} 
\usepackage{enumerate} 
\usepackage{upgreek} 
\usepackage{color} 
\usepackage{accents} 
\usepackage{mathtools} 
\usepackage{float} 
\usepackage{tikz}
\usetikzlibrary{matrix,graphs,arrows,positioning,calc,decorations.markings,shapes.symbols}
\usepackage[pdftex,bookmarks,colorlinks,breaklinks]{hyperref}  
\definecolor{dullmagenta}{rgb}{0.4,0,0.4}   
\definecolor{darkblue}{rgb}{0,0,0.4}
\hypersetup{linkcolor=red,citecolor=blue,filecolor=dullmagenta,urlcolor=darkblue} 
\usepackage{eucal}

\newtheorem{theorem}{Theorem}

\newtheorem{lemma}[theorem]{Lemma}

\newcommand{\Pain}[1]{\text{P}_{\mathrm{#1}}}
\newcommand{\dPain}[1]{\text{P}\left(\mathrm{#1}\right)}

\theoremstyle{definition}

\theoremstyle{remark}
\newtheorem{remark}[theorem]{Remark}

\numberwithin{equation}{section}

\begin{document}

{\noindent\Large\bf Gap Probabilities in the Laguerre Unitary Ensemble and Discrete Painlev\'e Equations}
\medskip
\begin{flushleft}

\textbf{Jie Hu}\\
Faculty of Science and Technology, Department of Mathematics, University of Macau, 
AE11 Avenida da Universidade, Taipa, Macau, China\\
E-mail: \href{mailto:hujie_0610@163.com}{\texttt{hujie\_0610@163.com}}\\[5pt] 

\textbf{Anton Dzhamay}\\
School of Mathematical Sciences, The University of Northern Colorado, Greeley, CO 80639, USA\\
E-mail: \href{mailto:anton.dzhamay@unco.edu}{\texttt{anton.dzhamay@unco.edu}}\\[5pt]

\textbf{Yang Chen}\\
Faculty of Science and Technology, Department of Mathematics, University of Macau, 
E11 Avenida da Universidade, Taipa, Macau, China\\
E-mail: \href{mailto:yayangchen@umac.mo}{\texttt{yayangchen@umac.mo}}\\[5pt] 



\end{flushleft}

\begin{flushright}\it
	To the memory of Jon Nimmo
\end{flushright}

\date{\today}

\begin{abstract}

In this paper we study a certain recurrence relation, that can be used to generate ladder operators for the 
Laguerre Unitary ensemble, from the point of view of Sakai's geometric theory of Painlev\'e equations. 
On one hand, this gives us one more detailed example of the appearance of discrete Painlev\'e equations
in the theory of orthogonal polynomials. On the other hand, it serves as a good illustration of the effectiveness 
of a recently proposed procedure on how to reduce such recurrences to some canonical discrete Painlev\'e
equations.

\end{abstract}

\section{Introduction} 
\label{sec:introduction}

By now it is clear that there are many fundamental connections between the theory of Random Matrices, Orthogonal Polynomials, 
and Painlev\'e Equations, both differential and discrete. Some conceptual understanding of this fact has been given in a series
of papers by Alexei Borodin and his collaborators \cite{AriBor:2006:MSDDPE,AriBor:2007:TDITP,Bor:2003:DGPADPE} and 
especially \cite{BorBoy:2003:DOTFPIDOPE}, see also a recent monograph of Walter Van Assche, \cite{Van:2018:OPAPE}. In a way, 
the geometric setting of Hidetaka Sakai's theory of Painlev\'e equations \cite{Sak:2001:RSAWARSGPE} seems to provide
the natural framework for questions involving the study of various orthogonal polynomial ensembles, and so it is not 
surprising that various objects of interest, such as the \emph{gap probabilities}, coefficients for \emph{three-term recurrence
relations}, or \emph{ladder operators}, can be described in terms of solution of either differential or discrete 
Painlev\'e equations. 

The purpose of the present paper is to study, from the geometric perspective of Sakai's theory, an example of a 
recurrence relation obtained by Shulin Lyu and Yang Chen in their study
of the largest eigenvalue distribution \cite{LyuChe:2017:LEDLUE} for the \emph{Laguerre Unitary Ensemble}, focusing on the 
reduction of this recurrence to a canonical form following step-by-step procedure recently proposed in \cite{DzhFilSto:2019:RCDOPHWDPE}.

Thus, we let the weight function be $w(x;\alpha) = x^{\alpha} \exp (-x)$, where $x>0$  and $\alpha>-1$ is a parameter and consider 
a family of monic polynomials 
\begin{equation}
	P_{j}(x,t) = x^{j} + p_{j-1}(t)x^{j-1} + \cdots + p_{0}(t)
\end{equation}
that are orthogonal with respect to the weight $w(x;\alpha)$ on the interval $[0,t]$, $0<t\leq\infty$, i.e.,
\begin{equation}
\int_{0}^{t} P_{n}(x,t) P_{m}(x,t) x^{\alpha} \exp (-x) dx = \delta_{n,m} h_{n}(t),	
\end{equation}
where $h_{n}(t)$ is the square of the $L^2\left([0,t];x^{\alpha} \exp (-x) dx\right)$ norm of $P_{n}(x,t)$. This 
unitary ensemble is called the \emph{Laguerre Unitary Ensemble} (or \emph{LUE} for short) since 
when $t=\infty$,  the family $\{P_{n}(x) = P_{n}(x,\infty)\}$ is the well-known family of monic Laguerre polynomials orthogonal 
w.r.t.~the weight $x^{\alpha} \exp (-x)$.

As usual, the orthogonality condition immediately implies the three term recurrence relations
\begin{equation*}
  x P_{n}(x,t) = P_{n+1} (x,t) + \alpha_{n}(t) P_{n}(x,t) + \beta_{n}(t) P_{n-1}(x,t),\qquad n\geq0,
\end{equation*}
with initial conditions $P_{-1}(x,t) =0$, $P_0(x,t)=1$.

The paper \cite{LyuChe:2017:LEDLUE} is concerned with the study of the probability $\mathbb{P}(n,t)$ that the largest 
eigenvalue in LUE on $[0,\infty)$ is not larger than $t$, where $n$ is the size of the corresponding random matrix. 
This probability can be computed as 
\begin{equation*}
	\mathbb{P}(n,t)=\frac{D_{n}(t)}{D_{n}(\infty)},
\end{equation*}
where 
\begin{equation*}
	D_{n}(t):=\det \left[ \int_{0}^{t} x^{j+k} x^{\alpha} \exp (-x) dx \right]_{0\leq j,k\leq n-1}
\end{equation*}
is the $n\times n$ \emph{Hankel determinant}, a fundamental object in the theory of orthogonal polynomials \cite{Sze:1967:OP},
that can be evaluated as $D_{n}(t) = \prod_{j=0}^{n-1} h_j(t)$, and 
\begin{equation*}
	D_n(\infty)= \frac{G(n+1)G(n+\alpha+1)}{\Gamma(\alpha+1)},\quad\text{where $G(.)$ is the Barnes $G$-function.}
\end{equation*}

One way to study (and generate) the family of orthogonal polynomials $\{P_{n}(x,t)\}$ is to use the lowering and raising 
ladder operators,
\begin{align*}
	\left(\frac{\textrm{d}}{\textrm{d}z}+B_n(z,t)\right)P_n(z,t)&= \beta_n(t) A_n(z,t) P_{n-1}(z,t), \\
	\left(\frac{\textrm{d}}{\textrm{d}z}-B_{n}(z,t)-v'(z)\right)P_{n-1}(z,t)&= -A_{n-1}(z,t)P_{n}(z,t),
\end{align*}
where $A_n(z,t)$ and $B_n(z,t)$ can be parameterized by the functions $R_{n}(t)$ and $r_{n}(t)$,
\begin{equation*}
	A_n(z,t)=\frac{R_n(t)}{z-t}+\frac{1-R_{n}(t)}{z},\qquad B_n(z,t)=\frac{r_n(t)}{z-t} - \frac{r_n(t)+n}{z},
\end{equation*}
and where 
\begin{equation*}
	R_n(t):= -\frac{P_n^2(t,t)}{h_n(t)}t^{\alpha}e^{-t},\qquad r_n(t):= -\frac{P_{n}(t,t)P_{n-1}(t,t)}{h_{n-1}(t)}t^{\alpha} e^{-t},\qquad
	\text{and }P_j(t,t):=P_j(z,t)|_{z=t}.
\end{equation*}

Let us now introduce a different parameterization $x_{n}(t), y_{n}(t)$ via
\begin{equation*}
	x_n(t)=1-\frac{1}{R_{n-1}(t)},\qquad y_n(t)= -r_n(t).
\end{equation*}
Then Lyu and Chen \cite[Remark~2.3]{LyuChe:2017:LEDLUE} showed that these variables satisfy the following recurrence relations in $n$:
\begin{equation}\label{eq:xyn-evol}
	\left\{ \begin{aligned}
		x_{n} x_{n+1} &= \frac{(y_{n} - n)(y_{n}-(n+ \alpha))}{y_{n}^{2}},\\
		y_{n} + y_{n-1} &= - \frac{(-t + 2n - 1 + \alpha) x_{n} - (2n - 1 + \alpha)}{(x_{n} - 1)^{2}}.
	\end{aligned}\right.
\end{equation}
This is the recurrence that we are interested in studying. We show, following 
the step-by-step procedure of \cite{DzhFilSto:2019:RCDOPHWDPE}, that this recurrence is a discrete Painlev\'e equation that is equivalent to one of the standard 
examples in the d-$\dPain{A_{3}^{(1)}/D_{5}^{(1)}}$ family. Our main result is the following Theorem.

\begin{theorem}\label{thm:coordinate-change}
	The recurrence \eqref{eq:xyn-evol} is equivalent to the standard discrete Painlev\'e  
	equation~\eqref{eq:dPD5-KNY} written in 
	\cite{KajNouYam:2017:GAOPE}. This equivalence is achieved 
	via the following change of variables:
	\begin{equation}\label{eq:xy2qp}
			x(q,p) = \frac{q(p+t) + a_{2}}{q p + a_{2}},\qquad
			y(q,p) = \frac{(p + t)(q p + a_{2})}{t}.
	\end{equation}
	The inverse change of variables is given by	
	\begin{equation}\label{eq:qp2xy}
		q(x,y) = \frac{(x-1)(x-1)y + n}{tx},\qquad
		p(x,y) = \frac{t(y - n)}{(x-1)y + n}.
	\end{equation} 
	The relationship between the Laguerre weight recurrence parameters and the root variables of discrete Painlev\'e equations is
	given by
	\begin{equation}
		a_{0} = n+\alpha,\quad a_{1} = - n,\quad a_{2} = n,\quad a_{3} = 1 - n - \alpha.
	\end{equation}
\end{theorem}

\begin{remark} Note that for our recurrence the root variables are constrained by the condition $a_{1} + a_{2} = 0$ (or, equivalently,
	$a_{0} + a_{3} = 1$).		
\end{remark}
	
These recurrences then are particular combinations of elementary mappings that can be thought of as B\"acklund transformations of a
differential $\Pain{V}$ equation that is associated with the same geometry. This is not surprising, since, 
if we put $\sigma_n(t):=t\frac{\textrm{d}}{\textrm{d}t}\ln{\mathbb{P}(n,t)}$, then it can be shown that it is the 
$\sigma$ function of a particular Painlev\'e V equation.
Estelle Basor and Yang Chen \cite{BasChe:2009:PDFDLSLUE} gave an alternate derivation 
of this result without relying on the Christoffel-Darboux kernel (or the reproducing kernel). Note also that 
the quantity $S_n(t) = 1- 1/R_n(t)$, satisfies
\begin{equation*}
S_n''(t)= \frac{3 S_n(t) -1}{2(S^2_n(t)-S_n(t))} (S_n'(t))^2 - \frac{S_n'(t)}{t}-\frac{\alpha^2(S_n(t)-1)^2}{2S_n(t)t^2}+(2n+1+\alpha)\frac{S_n(t)}{t}-\frac{1}{2}\frac{S_n(t)(S_n(t)+1)}{S_n(t)-1},	
\end{equation*}
which is a $\Pain{V}$ with parameters
\begin{equation*}
	\alpha_5=0, \qquad \beta_5= -\frac{\alpha^2}{2},\qquad c_5= 2n+1+\alpha,\qquad d_5=-1/2.
\end{equation*}
The function $\beta_n'(t) = t r_n'(t)$ satisfies a rather large second order non-linear ordinary differential equation in $t$,
and we will not reproduce it here.


\section{The Identification Procedure} 
\label{sec:the_identification_procedure}

\subsection{The Singularity Structure} 
\label{sub:the_singularity_structure}
To determine whether a given second-order nonlinear (non-autonomous) recurrence relation is one of discrete Painlev\'e equations, 
see the recent survey \cite{KajNouYam:2017:GAOPE}, the first step is to understand the singularity structure of the mapping
defined by this recurrence relation.
As is very common in this class of examples, our recurrence relation defines two natural mappings, 
the \emph{forward mapping $\psi_{1}^{(n)}: (x_{n},y_{n})\mapsto (x_{n+1},y_{n})$} defined by solving the first equation in
\eqref{eq:xyn-evol} for $x_{n+1}$ and the 
\emph{backward mapping $\psi_{2}^{(n)}:(x_{n},y_{n})\mapsto (x_{n},y_{n-1})$} defined by solving the second equation in
\eqref{eq:xyn-evol} for $y_{n-1}$. We are interested in studying the composed mapping 
$\psi^{(n)} = \left(\psi_{2}^{(n+1)}\right)^{-1} \circ \psi_{1}^{(n)}: (x_{n},y_{n})\mapsto (x_{n+1},y_{n+1})$. 
We put $x:=x_{n}$, $\overline{x}: = x_{n+1}$, 
$y:=y_{n}$, $\overline{y}: = y_{n+1}$ and sometimes omit the index $n$
in the mapping notation. The map $\psi:(x,y)\mapsto (\overline{x},\overline{y})$ then becomes
\begin{equation}\label{eq:fwd}
	\left\{
	\begin{aligned}
		\overline{x} &= \frac{(y - n)(y - (n+\alpha))}{x y^{2}},\\
		\overline{y} &= \frac{-y(y-n)^{2}(y - (n + \alpha))^{2} + x y^{2} (y - n)(y - (n+\alpha))(t + 2y - 2n - 1 - \alpha)
	+ x^{2} y^{4} (2n + 1 + \alpha - y)}{\left((y - n)(y - (n + \alpha)) - x y^{2}\right)^{2}}.
	\end{aligned}
	\right.
\end{equation}
Compactifying the mapping from $\mathbb{C}^{2}$ to $\mathbb{P}^{1} \times \mathbb{P}^{1}$ by introducing the 
coordinates $X = 1/x$ and $Y = 1/y$, it is easy to see that there are four affine base points of the mapping, and as
we see below, it is convenient to label them as
\begin{equation*}
	q_{1}(0,n),\quad q_{2}(0,n+\alpha),\quad q_{3}(1,\infty), \quad q_{7}(\infty,0);
\end{equation*}
(for example, it is immediate that at $q_{1}$ and $q_{2}$ both the numerator and the denominator of the mapping 
vanish, other points are found in the same way in other charts).
We resolve base point singularities using the blowup procedure, see, e.g., \cite{Sha:2013:BAG1}. That is,
for each base point $q_{i}(x_{i},y_{i})$ we construct two new local charts $(u_{i},v_{i})$ and $(U_{i},V_{i})$ 
given by $x = x_{i} + u_{i} = x_{i} + U_{i} V_{i}$ and $y = y_{i} + u_{i} v_{i} =y_{i} +  V_{i}$. 
The coordinates $v_{i} = 1/U_{i}$ represent all possible slopes of lines passing through the point $q_{i}$, 
and so this variable change ``separates'' all curves passing through $q_{i}$ based on their
slopes. This change of variables is a bijection away from $q_{i}$, but the point $q_{i}$ is 
replaced by the $\mathbb{P}^{1}$-line of all possible slopes, 
called the \emph{central fiber} or the \emph{exceptional divisor} of the blowup. 
We denote this central fiber by $F_{i}$, it is given in the blowup charts 
by local equations $u_{i}=0$ and $V_{i} = 0$. We then extend the mapping to these new charts via the above coordinate
substitution, find and resolve new base points (those would only appear on the exceptional divisors 
$u_{i} = V_{i} = 0$) and continue this process until it terminates (it should, in the discrete Painlev\'e case).
We summarize the result in the following Lemma.

\begin{lemma}
	\label{lem:base-pt-laguerre}
The base points of the mapping \eqref{eq:fwd}  
are 
\begin{alignat}{2}\label{eq-base-pt-laguerre}
	&q_{1}(x=0,y=n),&\quad &q_{2}(x=0,y=n+\alpha),\\
	&q_{3}\left(x=1,Y = \frac{1}{y}= 0\right)&\leftarrow 
	&q_{4}\left(u_{3} = x-1 = 0,v_{3}= \frac{1}{y(x-1)} = 0\right)\notag\\
	&&\leftarrow &q_{5}\left(U_{4}=y(1 - x)^{2} = t,V_{4} = \frac{1}{y(x-1)} = 0\right)\notag\\
	&&\leftarrow &q_{6}\left(U_{5}=(x-1)y((x-1)^{2} y - t) = t(1 - 2n + t - \alpha), V_{5} = \frac{1}{y(x-1)} = 0\right),\notag\\
	&q_{7}\left( X = \frac{1}{x} = 0, y= 0\right)&\leftarrow &q_{8}\left(U_{7} = \frac{1}{xy} = 0, V_{7} = y = 0\right).\notag
\end{alignat}
Considering the inverse mapping does not add any new base points.
\end{lemma}

Resolving these base points lifts our birational mapping 
$\psi: \mathbb{P}^{1} \times \mathbb{P}^{1} \dashrightarrow \mathbb{P}^{1} \times \mathbb{P}^{1}$ to the isomorphism, 
also denoted by $\psi$, between the corresponding algebraic surfaces, 
$\psi: \mathcal{X}_{\mathbf{b}}\to \mathcal{X}_{\overline{\mathbf{b}}}$. The subscript $\mathbf{b}$  indicates that 
the coordinates of the base points (and hence the resulting surface) depend on the
parameters of the mapping, $\mathbf{b}= \{\alpha, t, n\}$. These parameters can (and do) change
under the mapping and so $\overline{\mathbf{b}}$ denotes the evolved set of parameters. Sometimes we drop the 
parameters subscript and use the notation $\overline{\mathcal{X}}$ for the range of the mapping.


\subsection{The Induced Mapping on $\operatorname{Pic}(\mathcal{X})$} 
\label{sub:the_induced_mapping_on_operatorname_pic_mathcal_x}

The next step in the identification procedure is to compute the induced mapping on the Picard lattice.
Recall that for a regular algebraic variety $\mathcal{X}$, its \emph{Picard group} (or \emph{Picard lattice}) is the quotient 
of the \emph{divisor group} $\operatorname{Div}(\mathcal{X})=\operatorname{Span}_{\mathbb{Z}}(D)$ that is a free Abelian
group generated by closed irreducible subvarieties $D$ of codimension $1$, by the subgroup $\operatorname{P}(\mathcal{X})$ of 
\emph{principal divisors} (i.e., by the relation of \emph{linear equivalence}),
\begin{equation*}
	\operatorname{Pic}(\mathcal{X}) \simeq \operatorname{Cl}(\mathcal{X}) = \operatorname{Div}(\mathcal{X})/\operatorname{P}(\mathcal{X}) 
	= \operatorname{Div}(\mathcal{X})/\sim,
\end{equation*}
see \cite{SmiKahKek:2000:IAG} or \cite{Sha:2013:BAG1}. In our case, it is enough to know that 
$\operatorname{Pic}(\mathbb{P}^{1} \times \mathbb{P}^{1}) = \operatorname{Span}_{\mathbb{Z}}\{\mathcal{H}_{x},\mathcal{H}_{y}\}$, where 
$\mathcal{H}_{x} = [H_{x=a}]$ is the class of a \emph{vertical}  and  $\mathcal{H}_{y} = [H_{y=b}]$ is the class of a \emph{horizontal} 
line on $\mathbb{P}^{1}\times \mathbb{P}^{1}$. Each blowup procedure at a point $q_{i}$ adds the class $\mathcal{F}_{i} = [F_{i}]$
of the \emph{exceptional divisor} (i.e., the \emph{central fiber}) of the blowup, so 
$\operatorname{Pic}(\mathcal{X}) = \operatorname{Span}_{\mathbb{Z}}\{\mathcal{H}_{x},\mathcal{H}_{y},\mathcal{F}_{1},\ldots, \mathcal{F}_{8}\}$.
Further, the Picard lattice is equipped with the symmetric bilinear \emph{intersection form} given by 
\begin{equation}\label{eq:int-form}
\mathcal{H}_{x}\bullet \mathcal{H}_{x} = \mathcal{H}_{y}\bullet \mathcal{H}_{y} = \mathcal{H}_{x}\bullet \mathcal{F}_{i} = 
\mathcal{H}_{y}\bullet \mathcal{F}_{j} = 0,\qquad \mathcal{H}_{x}\bullet \mathcal{H}_{y} = 1,\qquad  \mathcal{F}_{i}\bullet \mathcal{F}_{j} = - \delta_{ij}	 
\end{equation}
on the generators, and then extended by the linearity. 

The mapping $\psi$ induces a linear mapping 
$\psi_{*}: \operatorname{Pic}(\mathcal{X}) \to \operatorname{Pic}(\overline{\mathcal{X}})$.
Note that $\operatorname{Pic}(\mathcal{X})$ and $\operatorname{Pic}(\overline{\mathcal{X}})$ are canonically isomorphic, 
so we sometimes just use the notation $\operatorname{Pic}(\mathcal{X})$. We also use $\overline{F}_{i}$ to denote the divisor 
of the central fiber of the blowup at the point 
$\overline{q}_{i} = \psi(q_{i})$, and similarly for the backwards mapping and for the classes;
notation $\mathcal{F}_{i\cdots j}$ stands for $\mathcal{F}_{i}+ \cdots + \mathcal{F}_{j}$.

\begin{lemma}\label{lem:dyn}
	The action of the mapping $\psi_{*}: \operatorname{Pic}(\mathcal{X})\to \operatorname{Pic}(\overline{\mathcal{X}})$
	is given by
	\begin{alignat*}{2}
		\mathcal{H}_{x}& \mapsto 5 \overline{\mathcal{H}}_{x} + 2 \overline{\mathcal{H}}_{y} - \overline{\mathcal{F}}_{12} - 
		2 \overline{\mathcal{F}}_{3456} - \overline{\mathcal{F}}_{78}, \qquad&
		\mathcal{H}_{y}& \mapsto 2\overline{\mathcal{H}}_{x} + \overline{\mathcal{H}}_{y} - \overline{\mathcal{F}}_{3456},  \\
		\mathcal{F}_{1} &\mapsto 2\overline{\mathcal{H}}_{x} + \overline{\mathcal{H}}_{y} -\overline{\mathcal{F}}_{23456}, \qquad&
		\mathcal{F}_{5} &\mapsto \overline{\mathcal{H}}_{x}  - \overline{\mathcal{F}}_{4}\\
		\mathcal{F}_{2} &\mapsto 2\overline{\mathcal{H}}_{x} + \overline{\mathcal{H}}_{y} -\overline{\mathcal{F}}_{13456},\qquad&
		\mathcal{F}_{6} &\mapsto \overline{\mathcal{H}}_{x}  - \overline{\mathcal{F}}_{3}\\
		\mathcal{F}_{3} &\mapsto \overline{\mathcal{H}}_{x}  - \overline{\mathcal{F}}_{6},\qquad&
		\mathcal{F}_{7} &\mapsto 2\overline{\mathcal{H}}_{x} + \overline{\mathcal{H}}_{y} -\overline{\mathcal{F}}_{34568}, \\
		\mathcal{F}_{4} &\mapsto \overline{\mathcal{H}}_{x}  - \overline{\mathcal{F}}_{5}, \qquad&
		\mathcal{F}_{8} &\mapsto 2\overline{\mathcal{H}}_{x} + \overline{\mathcal{H}}_{y} -\overline{\mathcal{F}}_{34567}.
	\end{alignat*}
	The evolution of parameters (and hence, the base points) is given by 
	$\mathbf{b}=\{\alpha, t, n\}\mapsto \overline{\mathbf{b}}=\{\alpha,t,n+1\}$.
\end{lemma}

\begin{proof} The proof of this Lemma is a standard direct computation and is omitted, see 
	\cite{DzhFilSto:2019:RCDOPHWDPE} or \cite{DzhTak:2018:OSAOSGTODPE} for similar examples worked out in detail.
\end{proof}

\subsection{The Surface Type} 
\label{sub:the_surface_type}
Given that our mapping is completely regularized by eight blowups, we know that it should fit into the discrete Painlev\'e equations framework.
To determine the type of the resulting algebraic surface, we need to find the configuration of the irreducible components of (the 
proper transform of) a bi-degree $(2,2)$ (or \emph{bi-quadratic}) curve $\Gamma$ on which these points lie. Since the proper transform of 
$\Gamma$ for a generic choice of parameters is the unique \emph{anti-canonical divisor} (i.e., the polar divisor of a symplectic form
$\omega$), we denote it by $-K_{\mathcal{X}}$. We also denote by $\eta$ the projection mapping back to $\mathbb{P}^{1} \times \mathbb{P}^{1}$,
\begin{equation*}
	\eta: \mathcal{X}_{\mathbf{b}} = \operatorname{Bl}_{q_{1}\cdots q_{8}}(\mathbb{P}^{1} \times \mathbb{P}^{1}) \to \mathbb{P}^{1} \times \mathbb{P}^{1}.
\end{equation*}

\begin{lemma}
Base points $q_{1},\ldots,q_{8}$ of the mapping \eqref{eq:fwd} lie on the bi-quadratic curve $\Gamma$ 
given in the affine chart by the equation $x=0$ (the homogeneous equation of $\Gamma$ is $x^{0}x^{1}y^{1}y^{1} = 0$, 
where $x = x^{0}/x^{1}$ and $y=y^{0}/y^{1}$, so $\Gamma$ is indeed bi-quadratic); note that some points come in
infinitely-close degeneration cascades. The irreducible components $d_{i}$
of the \emph{proper transform} $-K_{\mathcal{X}}$ of $\Gamma$,
\begin{equation*}
	-K_{\mathcal{X}} = 2 H_{x} + 2 H_{y} - F_{1} -\cdots - F_{8} = d_{0} + d_{1} + 2 d_{2} + 2 d_{3} + d_{4} + d_{5},
\end{equation*}
are given by 
\begin{equation}
	d_{0} = H_{x} - F_{1}-F_{2},\  d_{1} = H_{x}-F_{7}-F_{8},\  d_{2}=H_{y}-F_{3}-f_{4},\ 
	d_{3}= F_{4}-F_{5},\  d_{4}= F_{3} - F_{4},\  d_{5} = F_{5} - F_{6}; 
\end{equation}
they define the \emph{surface root basis} $\delta_{1},\ldots, \delta_{5}$ of $-2$-classes in $\operatorname{Pic}(\mathcal{X})$
whose configuration is described by the Dynkin diagram of type $D_{5}^{(1)}$:
\begin{figure}[H]
\begin{equation}\label{eq:d-roots-lw}
	\raisebox{-32.1pt}{\begin{tikzpicture}[
			elt/.style={circle,draw=black!100,thick, inner sep=0pt,minimum size=2mm}]
		\path 	(-1,1) 	node 	(d0) [elt, label={[xshift=-10pt, yshift = -10 pt] $\delta_{0}$} ] {}
		        (-1,-1) node 	(d1) [elt, label={[xshift=-10pt, yshift = -10 pt] $\delta_{1}$} ] {}
		        ( 0,0) 	node  	(d2) [elt, label={[xshift=-10pt, yshift = -10 pt] $\delta_{2}$} ] {}
		        ( 1,0) 	node  	(d3) [elt, label={[xshift=10pt, yshift = -10 pt] $\delta_{3}$} ] {}
		        ( 2,1) 	node  	(d4) [elt, label={[xshift=10pt, yshift = -10 pt] $\delta_{4}$} ] {}
		        ( 2,-1) node 	(d5) [elt, label={[xshift=10pt, yshift = -10 pt] $\delta_{5}$} ] {};
		\draw [black,line width=1pt ] (d0) -- (d2) -- (d1)  (d2) -- (d3) (d4) -- (d3) -- (d5);
	\end{tikzpicture}} \qquad
			\begin{alignedat}{2}
			\delta_{0} &= \mathcal{H}_{x} - \mathcal{F}_{1} - \mathcal{F}_{2}, &\qquad  \delta_{3} &= \mathcal{F}_{4} - \mathcal{F}_{5},\\
			\delta_{1} &= \mathcal{H}_{x} - \mathcal{F}_{7} - \mathcal{F}_{8}, &\qquad  \delta_{4} &= \mathcal{F}_{3} - \mathcal{F}_{4},\\
			\delta_{2} &= \mathcal{H}_{y} - \mathcal{F}_{3} - \mathcal{F}_{4}, &\qquad  \delta_{5} &= \mathcal{F}_{5} - \mathcal{F}_{6}.
			\end{alignedat}
\end{equation}
	\caption{The Surface Root Basis for the Laguerre Weight Recurrence}
	\label{fig:d-roots-lw}
\end{figure}
\end{lemma}

\begin{figure}[ht]
	\begin{tikzpicture}[>=stealth,basept/.style={circle, draw=red!100, fill=red!100, thick, inner sep=0pt,minimum size=1.2mm}]
		\begin{scope}[xshift = 0cm]
			\draw [black, line width = 1pt]  	(4,0) 	-- (-0.5,0) 	node [left] {$H_{y}$} node[pos=0, right] {$y=0$};
			\draw [black, line width = 1pt] 	(4,2.5) -- (-0.5,2.5)	node [left] {$H_{y}$} node[pos=0, right] {$y=\infty$};
			\draw [black, line width = 1pt] 	(0,3) -- (0,-0.5)		node [below] {$H_{x}$} node[pos=0, above, xshift=-7pt] {$x=0$};
			\draw [black, line width = 1pt] 	(3,3) -- (3,-0.5)		node [below] {$H_{x}$} node[pos=0, above, xshift=7pt] {$x=\infty$};
			
			\node (q1) at (0,0.5) 	[basept,label={[left] $q_{1}$}] {};
			\node (q2) at (0,1.3) 	[basept,label={[left] $q_{2}$}] {};
			\node (q3) at (0.6,2.5) [basept,label={[below left] $q_{3}$}] {};
			\node (q4) at (0.8,3.0) [basept,label={[above ] $q_{4}$}] {};
			\node (q5) at (1.5,3.0) 	[basept,label={[above ] $q_{5}$}] {};
			\node (q6) at (2.2,3.0) 	[basept,label={[above ] $q_{6}$}] {};
			\node (q7) at (3,0)		[basept,label={[above right] $q_{7}$}] {};
			\node (q8) at (3.5,-0.5) [basept,label={[below right] $q_{8}$}] {};
			\draw [line width = 0.8pt, ->] (q6) edge (q5) (q5) edge (q4) (q4) edge (q3);
			\draw [line width = 0.8pt, ->] (q8) edge (q7);
		\end{scope}

		\draw [->] (7,1.5)--(5,1.5) node[pos=0.5, below] {$\operatorname{Bl}_{q_{1}q_{2}q_{3}q_{7}}$};

		\begin{scope}[xshift = 9cm]
			\draw [red, line width = 1pt] 	(-0.5,2.5) -- (4,2.5)	node [right] {$H_{y} - F_{3}$};
			\draw [blue, line width = 1pt] 	(0,3) -- (0,-0.5)		node [below] {$H_{x}-F_{12}$};
			\draw [red, line width = 1pt ] (-0.5,0) .. controls (2,0) and (2.5,-0.5) .. (3,-1) node [below] {$H_{y} - F_{7}$};
			\draw [red, line width = 1pt] 	(3.5,3)  .. controls (3.5,1.5) and (3.7,1) .. (4,0.5)  node [right] {$H_{x} - F_{7}$};

			\draw [red, line width = 1 pt] (-0.4,0.7) -- (0.4,0.3) node [pos = 0, left] {$F_{1}$};
			\draw [red, line width = 1 pt] (-0.4,1.5) -- (0.4,1.1) node [pos = 0, left] {$F_{2}$};
			\draw [red, line width = 1 pt] (0.4,2.9) -- (1,1.7)  node [below right] {$F_{3}$};
			\draw [red, line width = 1 pt] (2,-1) -- (4,1) node [above right] {$F_{7}$};

			\node (q4) at (0.6,2.5) [basept,label={[below left] $q_{4}$}] {};
			\node (q5) at (0.8,3.0) [basept,label={[above ] $q_{5}$}] {};
			\node (q6) at (1.5,3.0) [basept,label={[above ] $q_{6}$}] {};
			\node (q8) at (3.825,0.83) [basept,label={[xshift=-8pt,yshift=-6pt] $q_{8}$}] {};
			
			\draw [line width = 0.8pt, ->] (q6) edge (q5) (q5) edge (q4);
			
		\end{scope}

		\draw [->] (10.5,-2.7)--(10.5,-1.2) node[pos=0.5, right] {$\operatorname{Bl}_{q_{4}q_{8}}$};

		\begin{scope}[xshift = 9cm, yshift = -6.5cm]
			\draw [blue, line width = 1pt] 	(-0.5,2.5) -- (4,2.5)	node [right] {$H_{y} - F_{34}$};
			\draw [blue, line width = 1pt] 	(0,3) -- (0,-0.5)		node [below] {$H_{x}-F_{12}$};
			\draw [red, line width = 1pt ] (-0.5,0) .. controls (2,0) and (2.5,-0.5) .. (3,-1) node [below] {$H_{y} - F_{7}$};
			\draw [blue, line width = 1pt] 	(3.5,3)  .. controls (3.5,1.5) and (3.7,1) .. (4,0.5)  node [right] {$H_{x} - F_{78}$};

			\draw [red, line width = 1 pt] (-0.4,0.7) -- (0.4,0.3) node [pos = 0, left] {$F_{1}$};
			\draw [red, line width = 1 pt] (-0.4,1.5) -- (0.4,1.1) node [pos = 0, left] {$F_{2}$};
			\draw [red, line width = 1 pt] (0.4,2.9) -- (1.6,0.5)  node [pos = 0, above right] {$F_{4}$};
			\draw [blue, line width = 1 pt] (0.7,1.3) .. controls (2,1.3) and (2.5,1.1) .. (2.5,0.5)  node [below] {$F_{3} - F_{4}$};
			\draw [red, line width = 1 pt] (4.0,-0.5) .. controls (3.5,0.1) and (3,0.7) .. (4.5,1.3) node [right] {$F_{8}$};
			\draw [blue, line width = 1pt] 	(2,-1)  .. controls (2.5,-0.7) and (3,-0.3) .. (4,0)  node [right] {$F_{7} - F_{8}$};

			\node (q5) at (0.9,1.9) [basept,label={[above right] $q_{5}$}] {};
			\node (q6) at (1.5,1.9) [basept,label={[above right] $q_{6}$}] {};
			
			\draw [line width = 0.8pt, ->] (q6) edge (q5);
			
		\end{scope}

		\draw [->] (5.5,-5.5)--(7.5,-5.5) node[pos=0.5, below] {$\operatorname{Bl}_{q_{7}q_{8}}$};

		\begin{scope}[xshift = 0cm, yshift = -6.5cm]
			\draw [blue, line width = 1pt] 	(-0.5,2.5) -- (4,2.5)	node [right] {$H_{y} - F_{34}$};
			\draw [blue, line width = 1pt] 	(0,3) -- (0,-0.5)		node [below] {$H_{x}-F_{12}$};
			\draw [red, line width = 1pt ] (-0.5,0) .. controls (2,0) and (2.5,-0.5) .. (3,-1) node [below] {$H_{y} - F_{7}$};
			\draw [blue, line width = 1pt] 	(3.5,3)  .. controls (3.5,1.5) and (3.7,1) .. (4,0.5)  node [right] {$H_{x} - F_{78}$};

			\draw [red, line width = 1 pt] (-0.4,0.7) -- (0.4,0.3) node [pos = 0, left] {$F_{1}$};
			\draw [red, line width = 1 pt] (-0.4,1.5) -- (0.4,1.1) node [pos = 0, left] {$F_{2}$};
			\draw [blue, line width = 1 pt] (0.4,2.9) -- (1.6,0.5)  node [pos = 0, above right] {$F_{4}-F_{5}$};
			\draw [blue, line width = 1 pt] (0.3,1.8) -- (2.8,1.8)  node [above] {$F_{5}-F_{6}$};
			\draw [red, line width = 1 pt] (2.0,2) -- (2.8,1.2)  node [below right] {$F_{6}$};
			\draw [blue, line width = 1 pt] (0.7,1.3) .. controls (2,1.3) and (2.5,1.1) .. (2.5,0.5)  node [below] {$F_{3} - F_{4}$};
			\draw [red, line width = 1 pt] (4.0,-0.5) .. controls (3.5,0.1) and (3,0.7) .. (4.5,1.3) node [right] {$F_{8}$};
			\draw [blue, line width = 1pt] 	(2,-1)  .. controls (2.5,-0.7) and (3,-0.3) .. (4,0)  node [right] {$F_{7} - F_{8}$};
		\end{scope}

		\draw [->] (1.5,-2.7)--(1.5,-1.2) node[pos=0.5, right] {$\operatorname{Bl}_{q_{1}\cdots q_{8}}$};

	\end{tikzpicture}
	\caption{The Sakai Surface for the Laguerre Weight Recurrence ($F_{i\cdots j} = F_{i} + \cdots + F_{j}$).}
	\label{fig:surface-laguerre}
\end{figure}
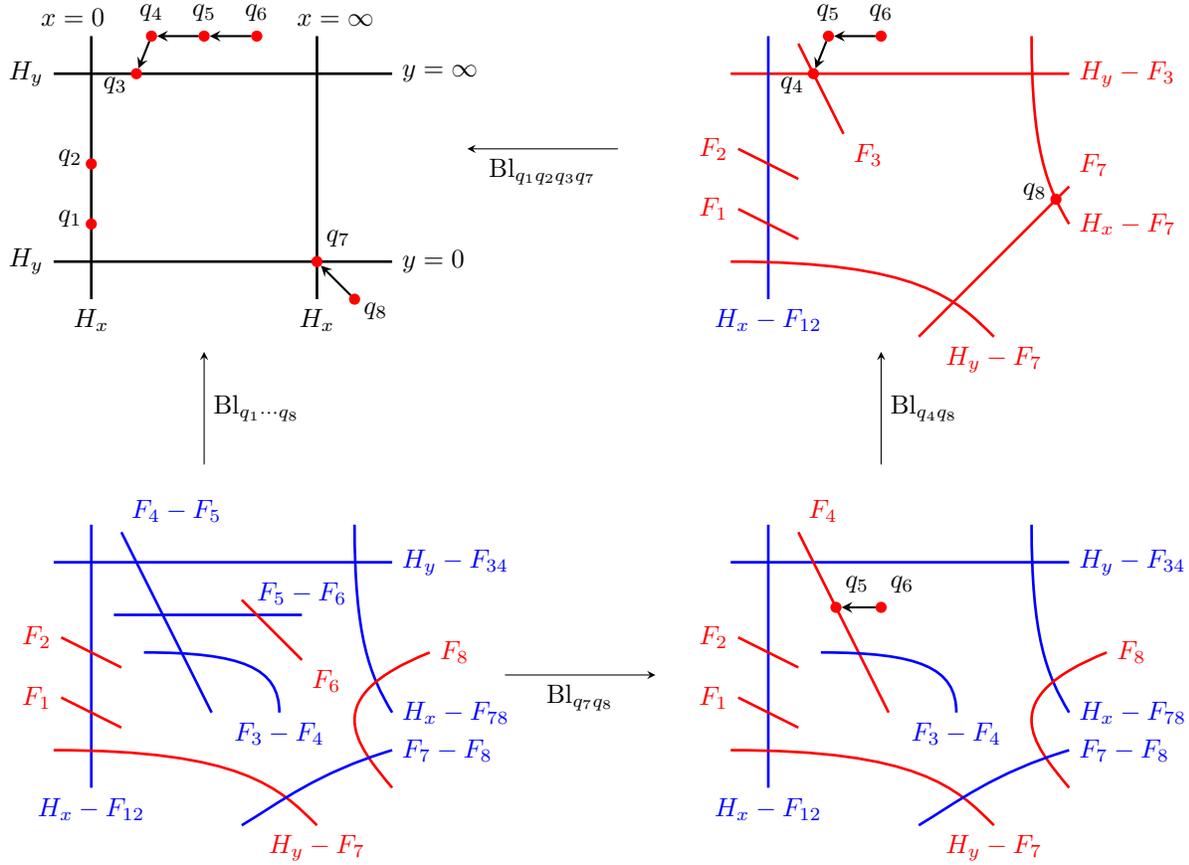

We show some intermediate stages of the blowup process and the resulting $D_{5}^{(1)}$ surface on Figure~\ref{fig:surface-laguerre}. 
Thus our recurrence belongs to the d-$\dPain{A_{3}^{(1)}/D_{5}^{(1)}}$ family with the symmetry group 
$\widetilde{W}\left(A_{3}^{(1)}\right)$. We describe the choice of the standard d-$\dPain{A_{3}^{(1)}/D_{5}^{(1)}}$
point configuration, choices of the root bases for the surface and the symmetry sub-lattices, and other data, in the Appendix; 
we follow \cite{KajNouYam:2017:GAOPE} in our conventions.

\begin{remark} Looking at Figure~\ref{fig:surface-laguerre}, we notice another $-2$-curve $F_{7}-F_{8}$ that is disjoint from the 
	irreducible components of the anti-canonical divisor. Such curves form the class $\Delta^{\text{nod}}$, see 
	\cite[Section 3.3]{Sak:2001:RSAWARSGPE}
	
\end{remark}

\subsection{Initial Geometry Identification} 
\label{sub:initial_geometry_identification}
\label{sub:subsection_name}
The next step in the identification process is to find some change of basis in $\operatorname{Pix}(\mathcal{X})$ from the basis 
$\{\mathcal{H}_{x},\mathcal{H}_{y},\mathcal{F}_{i}\}$ to the basis $\{\mathcal{H}_{q},\mathcal{H}_{p},\mathcal{E}_{j}\}$ that 
correspond to the standard geometry configuration that identifies the surface root bases; we refer to this step as \emph{matching the geometry}.
At this point there are many possible choices of such basis change, we later may have to adjust it to \emph{match the dynamics}.

\begin{lemma}\label{lem:change-basis-pre}
	The following change of basis of $\operatorname{Pic}(\mathcal{X})$ identifies the root bases 
	between the standard $D_{5}^{(1)}$ surface and the surface that we obtained for the Laguerre weight recurrence:
	\begin{align*}
		\mathcal{H}_{x} & = \mathcal{H}_{p}, & \qquad 
		\mathcal{H}_{q} & = 2\mathcal{H}_{x} + \mathcal{H}_{y} - \mathcal{F}_{1} - \mathcal{F}_{3} - \mathcal{F}_{4} - \mathcal{F}_{7},\\
		\mathcal{H}_{y} & = \mathcal{H}_{q} + 2\mathcal{H}_{p}  - \mathcal{E}_{1} -\mathcal{E}_{3} - \mathcal{E}_{5} - \mathcal{E}_{6}, & 
		\qquad 	\mathcal{H}_{p} &= \mathcal{H}_{x},\\
		\mathcal{F}_{1} & = \mathcal{H}_{p} - \mathcal{E}_{1}, & \qquad \mathcal{E}_{1} & = \mathcal{H}_{x} - \mathcal{F}_{1},\\
		\mathcal{F}_{2} & = \mathcal{E}_{2}, & \qquad \mathcal{E}_{2} & = \mathcal{F}_{2},\\
		\mathcal{F}_{3} & = \mathcal{H}_{p} - \mathcal{E}_{6}, & \qquad \mathcal{E}_{3} & = \mathcal{H}_{x} - \mathcal{F}_{7},\\
		\mathcal{F}_{4} & = \mathcal{H}_{p} - \mathcal{E}_{5}, & \qquad \mathcal{E}_{4} & = \mathcal{F}_{8},\\
		\mathcal{F}_{5} & = \mathcal{E}_{7}, & \qquad \mathcal{E}_{5} & = \mathcal{H}_{x} - \mathcal{F}_{4},\\
		\mathcal{F}_{6} & = \mathcal{E}_{8}, & \qquad \mathcal{E}_{6} & = \mathcal{H}_{x} - \mathcal{F}_{3},\\
		\mathcal{F}_{7} & = \mathcal{H}_{p} - \mathcal{E}_{3}, & \qquad \mathcal{E}_{7} & = \mathcal{F}_{5},\\
		\mathcal{F}_{8} & = \mathcal{E}_{4}, & \qquad \mathcal{E}_{8} & = \mathcal{F}_{6}.
		\end{align*}
\end{lemma}

\begin{proof}
	This is a direct computation based on comparing the surface root bases on Figure~\ref{fig:d-roots-lw} and  Figure~\ref{fig:d-roots-d51}. 
\end{proof}


\subsection{The Symmetry Roots and the Translations} 
\label{sub:the_symmetry_roots_and_the_translations}
We are now in a position to compare the dynamics. Note that there are two \emph{non-equivalent} model examples of discrete Painlev\'e
equations, that we label as $\mathbf{[\overline{1}1 \overline{1}1]}$ and $\mathbf{[\overline{1}001]}$,
on the $D_{5}^{(1)}$-surface that are described in Section~\ref{sub:some_standard_discrete_d_dpain_a__3_1_d__5_1_equations} in the 
Appendix. It is interesting that the mapping \eqref{eq:xyn-evol} has the \emph{multiplicative-additive} from that looks
very similar to the mapping \eqref{eq:dPD5-Sakai}, but instead it is equivalent to the mapping~\eqref{eq:dPD5-KNY} that has the purely
\emph{additive} form. To show that, we start with the standard choice of the symmetry root basis \eqref{eq:a-roots-a3}			
and use the change of basis in Lemma~\ref{lem:change-basis-pre} to get the symmetry roots for the applied problem shown on 
Figure~\ref{fig:a-roots-lw-pre}.	
\begin{figure}[ht]
\begin{equation}\label{eq:a-roots-lw-pre}			
	\raisebox{-32.1pt}{\begin{tikzpicture}[
			elt/.style={circle,draw=black!100,thick, inner sep=0pt,minimum size=2mm}]
		\path 	(-1,1) 	node 	(a0) [elt, label={[xshift=-10pt, yshift = -10 pt] $\alpha_{0}$} ] {}
		        (-1,-1) node 	(a1) [elt, label={[xshift=-10pt, yshift = -10 pt] $\alpha_{1}$} ] {}
		        ( 1,-1) node  	(a2) [elt, label={[xshift=10pt, yshift = -10 pt] $\alpha_{2}$} ] {}
		        ( 1,1) 	node 	(a3) [elt, label={[xshift=10pt, yshift = -10 pt] $\alpha_{3}$} ] {};
		\draw [black,line width=1pt ] (a0) -- (a1) -- (a2) --  (a3) -- (a0); 
	\end{tikzpicture}} \qquad
			\begin{alignedat}{2}
			\alpha_{0} &= \mathcal{F}_{1} - \mathcal{F}_{2}, &\qquad  \alpha_{2} &= \mathcal{F}_{7} - \mathcal{F}_{8},\\
			\alpha_{1} &= \mathcal{H}_{y} - \mathcal{F}_{1} - \mathcal{F}_{7}, &\qquad  \alpha_{3} &= 2\mathcal{H}_{x} + \mathcal{H}_{y} - 
			\mathcal{F}_{1} - \mathcal{F}_{3} - \mathcal{F}_{4} - \mathcal{F}_{5} - \mathcal{F}_{6} - \mathcal{F}_{7}.
			\\[5pt]
			\delta & = \mathrlap{\alpha_{0} + \alpha_{1} +  \alpha_{2} + \alpha_{3}.} 
			\end{alignedat}
\end{equation}
	\caption{The Symmetry Root Basis for the Laguerre Weight Recurrence (preliminary choice)}
	\label{fig:a-roots-lw-pre}	
\end{figure}
From the action of $\psi_{*}$ on $\operatorname{Pic}(\mathcal{X})$ given in Lemma~\ref{lem:dyn} we 
can now obtain the corresponding translation on the root lattice, decompose it in terms of the generators 
of the extended affine Weyl symmetry group, and compare the results with the standard mappings $\varphi$ and $\phi$ 
given in Section~\ref{sub:some_standard_discrete_d_dpain_a__3_1_d__5_1_equations}. We get
\begin{alignat*}{2}
	\psi_{*}: \upalpha &=  \langle \alpha_{0}, \alpha_{1}, \alpha_{2}, \alpha_{3}  \rangle
	\mapsto \psi_{*}(\upalpha) = \upalpha + \langle 0,-1,0,1 \rangle \delta,&\qquad 
	\psi &= \sigma_{3}\sigma_{2}w_{1}w_{2}w_{0}w_{1},\\
	\varphi_{*}: \upalpha &=  \langle \alpha_{0}, \alpha_{1}, \alpha_{2}, \alpha_{3} \rangle
	\mapsto \varphi_{*}(\upalpha) = \upalpha + \langle -1,1,-1,1 \rangle \delta,&\qquad
	\varphi &= \sigma_{3}\sigma_{2}w_{3}w_{1}w_{2}w_{0},\\
	\phi_{*}: \upalpha &=  \langle \alpha_{0}, \alpha_{1}, \alpha_{2}, \alpha_{3} \rangle
	\mapsto \varphi_{*}(\upalpha) = \upalpha + \langle -1,0,0,1 \rangle \delta, &\qquad 
	\phi &= \sigma_{3}\sigma_{1}w_{2}w_{1}w_{0}.
\end{alignat*}

From here
we immediately see that $\psi = w_{1}\circ \varphi \circ w_{1}^{-1}$ (note that $w_{1} \sigma_{3}\sigma_{2}  = \sigma_{3}\sigma_{2} w_{3}$ 
and that $w_{1}$ is an involution, $w_{1}^{-1} = w_{1}$). Thus, our dynamic is equivalent to the standard equation~\eqref{eq:dPD5-KNY} written in 
\cite{KajNouYam:2017:GAOPE} but is different from equation~\eqref{eq:dPD5-Sakai} written in \cite{Sak:2001:RSAWARSGPE}
(i.e., $\mathbf{[0\overline{1}01]} = \mathbf{[\overline{1}1 \overline{1}1]}$). To find the change of variables
matching the two equations we first need to adjust our change of basis in $\operatorname{Pic}(\mathcal{X})$ by acting on it by $w_{1}$,
so that we match not only the geometry, but also the \emph{dynamics}. We do it in the next section.

\subsection{Final Geometry Identification} 
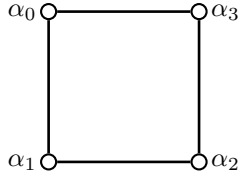
\label{sub:final_geometry_identification}

\begin{figure}[ht]
\begin{equation}\label{eq:a-roots-lw-fin}			
	\raisebox{-32.1pt}{\begin{tikzpicture}[
			elt/.style={circle,draw=black!100,thick, inner sep=0pt,minimum size=2mm}]
		\path 	(-1,1) 	node 	(a0) [elt, label={[xshift=-10pt, yshift = -10 pt] $\alpha_{0}$} ] {}
		        (-1,-1) node 	(a1) [elt, label={[xshift=-10pt, yshift = -10 pt] $\alpha_{1}$} ] {}
		        ( 1,-1) node  	(a2) [elt, label={[xshift=10pt, yshift = -10 pt] $\alpha_{2}$} ] {}
		        ( 1,1) 	node 	(a3) [elt, label={[xshift=10pt, yshift = -10 pt] $\alpha_{3}$} ] {};
		\draw [black,line width=1pt ] (a0) -- (a1) -- (a2) --  (a3) -- (a0); 
	\end{tikzpicture}} \qquad
			\begin{alignedat}{2}
			\alpha_{0} &= \mathcal{H}_{y} - \mathcal{F}_{2} - \mathcal{F}_{7}, &\qquad  
			\alpha_{2} &= \mathcal{H}_{y} -  \mathcal{F}_{1} - \mathcal{F}_{8},\\
			\alpha_{1} &= \mathcal{F}_{1} + \mathcal{F}_{7} - \mathcal{H}_{y}, &\qquad  
			\alpha_{3} &= 2\mathcal{H}_{x} + \mathcal{H}_{y} - 
			\mathcal{F}_{1} - \mathcal{F}_{3} - \mathcal{F}_{4} - \mathcal{F}_{5} - \mathcal{F}_{6} - \mathcal{F}_{7}.
			\\[5pt]
			\delta & = \mathrlap{\alpha_{0} + \alpha_{1} +  \alpha_{2} + \alpha_{3}.} 
			\end{alignedat}
\end{equation}
	\caption{The Symmetry Root Basis for the Laguerre Weight Recurrence (final choice)}
	\label{fig:a-roots-lw-fin}	
\end{figure}

\begin{lemma}\label{lem:change-basis-fin}
	After the change of basis of $\operatorname{Pic}(\mathcal{X})$ given by
	\begin{align*}
		\mathcal{H}_{x} & = \mathcal{H}_{q} + \mathcal{H}_{p} - \mathcal{E}_{5} - \mathcal{E}_{6}, & \qquad 
		\mathcal{H}_{q} & = 2\mathcal{H}_{x} + \mathcal{H}_{y} - \mathcal{F}_{1} - \mathcal{F}_{3} - \mathcal{F}_{4} - \mathcal{F}_{7},\\
		\mathcal{H}_{y} & = \mathcal{H}_{q} + 2\mathcal{H}_{p}  - \mathcal{E}_{1} -\mathcal{E}_{3} - \mathcal{E}_{5} - \mathcal{E}_{6}, & 
		\qquad 	\mathcal{H}_{p} &= \mathcal{H}_{x} + \mathcal{H}_{y} - \mathcal{F}_{1} - \mathcal{F}_{7},\\
		\mathcal{F}_{1} & = \mathcal{H}_{q} + \mathcal{H}_{p} - \mathcal{E}_{1} - \mathcal{E}_{5} - \mathcal{E}_{6}, & \qquad 
		\mathcal{E}_{1} & = \mathcal{H}_{x} - \mathcal{F}_{1},\\
		\mathcal{F}_{2} & = \mathcal{E}_{2}, & \qquad \mathcal{E}_{2} & = \mathcal{F}_{2},\\
		\mathcal{F}_{3} & = \mathcal{H}_{p} - \mathcal{E}_{6}, & \qquad \mathcal{E}_{3} & = \mathcal{H}_{x} - \mathcal{F}_{7},\\
		\mathcal{F}_{4} & = \mathcal{H}_{p} - \mathcal{E}_{5}, & \qquad \mathcal{E}_{4} & = \mathcal{F}_{8},\\
		\mathcal{F}_{5} & = \mathcal{E}_{7}, & \qquad 
		\mathcal{E}_{5} & = \mathcal{H}_{x} + \mathcal{H}_{y} - \mathcal{F}_{1} - \mathcal{F}_{4} - \mathcal{F}_{7},\\
		\mathcal{F}_{6} & = \mathcal{E}_{8}, & \qquad 
		\mathcal{E}_{6} & =  \mathcal{H}_{x} + \mathcal{H}_{y} - \mathcal{F}_{1} - \mathcal{F}_{3} - \mathcal{F}_{7},\\
		\mathcal{F}_{7} & = \mathcal{H}_{q} + \mathcal{H}_{p} - \mathcal{E}_{3} - \mathcal{E}_{5} - \mathcal{E}_{6}, & \qquad 
		\mathcal{E}_{7} & = \mathcal{F}_{5},\\
		\mathcal{F}_{8} & = \mathcal{E}_{4}, & \qquad \mathcal{E}_{8} & = \mathcal{F}_{6}.
		\end{align*}
		the recurrence relations \eqref{eq:xyn-evol}
		for variables $x_{n}$ and $y_{n}$ coincides with the discrete Painlev\'e equation given by
		\eqref{eq:dPD5-KNY}. The resulting identification of the symmetry root bases (the surface root bases do not change) is shown in 
		Figure~\ref{fig:a-roots-lw-fin}.
\end{lemma}

Next we need to realize this change of basis on $\operatorname{Pic}(\mathcal{X})$ by an explicit change of coordinates. For that, it is convenient to first 
match the parameters between the applied problem and the reference example. This is done with the help of the \emph{Period Map}.

\subsection{The Period Map and the Identification of Parameters} 
\label{sub:the_period_map_and_the_identification_of_parameters}

For the root variable parameterization, let us consider a generic point configuration corresponding to the geometry of Figure~\ref{fig:surface-laguerre}.
Using the action of the $\mathbf{PGL}_{2}(\mathbb{C})\times \mathbf{PGL}_{2}(\mathbb{C})$ gauge group 
we can put $H_{x} - F_{1} - F_{2} = V(x)$, $H_{x} - F_{7} - F_{8} = V(X)$, $H_{y} - F_{3} - F_{4} = V(Y)$, and $q_{7}(\infty,0)$. This leaves the scale
freedom on the coordinates $x$ and $y$; we use the scaling in the $x$-coordinate to put $q_{3}(1,\infty)$. Then our point configuration can be described in 
terms of generic parameters $c_{i}$ as 
\begin{equation*}
	q_{1}(0,c_{1}),\ q_{2}(0,c_{2}),\ q_{3}(1,\infty)\leftarrow q_{4}(u_{3}=0,v_{3}=0)\leftarrow q_{5}(c_{5},0)\leftarrow q_{6}(c_{6},0),\ 
	q_{7}(\infty,0)\leftarrow q_{8}(U_{8} = 0,V_{8} = 0)	
\end{equation*}
with the remaining scaling gauge action in the $y$-coordinate given by
\begin{equation*}
		\left(\begin{matrix}
			c_{1} & c_{2}\\
			c_{5} & c_{6} 
		\end{matrix};  \begin{matrix}
			x \\ y
		\end{matrix}\right) \sim \left(\begin{matrix}
			\lambda c_{1} & \lambda c_{2}\\
			\lambda c_{5} & \lambda^{2} c_{6} 
		\end{matrix}; \begin{matrix}
			x\\ \lambda y
		\end{matrix}\right),\quad \lambda\neq0.
\end{equation*}

It is immediate that the points $q_{i}$ lie on the polar divisor of a symplectic form given in the affine 
$(x,y)$ chart by $\omega = k \frac{dx\wedge dy}{x}$. We then have the following Lemma.

\begin{lemma}
	\qquad
	
	\begin{enumerate}[(i)]
		\item The residues of the symplectic form $\omega = k \frac{dx\wedge dy}{x}$
		along the irreducible components of the polar divisor are given by
		\begin{alignat*}{3}
			\operatorname{res}_{d_{0}} \omega &=  k dy, &\qquad
			\operatorname{res}_{d_{2}} \omega &=  0,\quad &\qquad 
			\operatorname{res}_{d_{4}} \omega &=  -k \frac{dv_{3}}{v_{3}^{2}},\\
			\operatorname{res}_{d_{1}} \omega &=  -k dy,&\qquad 
			\operatorname{res}_{d_{3}} \omega &=  -3k\, dU_{4}, &\qquad
			\operatorname{res}_{d_{5}} \omega &=  k \frac{3\, dU_{5}}{c_{6}}.
		\end{alignat*}
		
		\item The root variables are given by 
		\begin{equation}\label{eq:root-vars-lw}
			a_{0} = - k c_{2}, \qquad a_{1} = k c_{1},\qquad 
			a_{2} = - k c_{1}, \qquad a_{3} = k\left(- c_{1} + c_{5} - \frac{c_{6}}{c_{5}}\right),
		\end{equation}
		and so the root variables are constrained by $a_{1} + a_{2} = 0$. 
		Without loss of generality we can put $k=-1$ and then use the $\lambda$ gauge scaling to ensure the standard normalization condition 
		$a_{0} + a_{1} + a_{2} + a_{3} =1$. Then we get 
		\begin{equation}\label{eq:par-match-app}
			c_{2} = a_{0},\qquad c_{1} = -a_{1} = a_{2}, \qquad c_{6} = c_{5}(1 - c_{1} - c_{2} + c_{5}),
		\end{equation}
		which shows that the application parameters are in fact generic for this point configuration; putting 
		$n = a_{2}$ and $\alpha = a_{0} - a_{2}$, as well as denoting $c_{6}$ by $t$, establishes this equivalence. Note that the 
		parameter evolution is now consistent between the root variables and the application parameters; $\overline{n} = n+1$, 
		$\overline{\alpha} = \alpha$, and $\overline{t} = t$.
		\end{enumerate}
\end{lemma}


\subsection{The Change of Coordinates} 
\label{sub:the_change_of_coordinates}
We are now ready to prove Theorem~\ref{thm:coordinate-change}. Note that at this point we have not shown that the parameter $t$
in \eqref{eq:dPD5-KNY} is the same as in \eqref{eq:xyn-evol}, so we continue working with generic parameters $c_{i}$   
from the previous section.

\begin{proof}(\textbf{Theorem 5}) The proof is standard, and so we only outline the key steps. From the linear change of basis on 
	$\operatorname{Pic}(\mathcal{X})$ given in Lemma~\ref{lem:change-basis-fin}, we see that $x$ is a projective coordinate
	on a pencil of $(1,1)$ curves in the $(q,p)$-plane passing through the points $p_{5}$ and $p_{6}$, and $y$ is a projective 
	coordinate on a pencil of $(1,2)$ curves in the $(q,p)$-plane passing through the points $p_{1}$, $p_{3}$, $p_{5}$, and $p_{6}$.
	The bases for these pencils are given by the curves with affine defining polynomials $\{q,qp - a_{1}\}$ and $\{1, p(q(p+t) - a_{1})\}$,
	i.e., 
	\begin{equation*}
		x = \frac{A q + B(qp - a_{1})}{C q + D(qp - a_{1})},\qquad y = \frac{K + Lp(q(p+t) - a_{1})}{M + Np(q(p+t) - a_{1})}.
	\end{equation*}
	Using the correspondence between the exceptional divisor classes for $F_{i}$, $i=1,2,3,4,7,8$ 
	allows us to fix the values of the coefficients $A,\ldots,N$ to get $x = \frac{q(p+t)+a_{2}}{qp + a_{2}}$ and $y = \frac{(p+t)(qp + a_{2})}{t}$.
	Moreover, the correspondence $F_{7} - F_{8} = H_{q} + H_{p} - E_{3} - E_{4} - E_{5} - E_{7}$ imposes the $a_{1} + a_{2} = 0$ constraint,
	and the condition that $F_{5} - F_{6} = F_{7} - F_{8}$ shows that $c_{5} = t$, as expected. The inverse change of variables is obtained 
	along the same lines.	
\end{proof}



\appendix

\section{Discrete Painlev\'e Equations in the \lowercase{d}-$\dPain{A_{3}^{(1)}/D_{5}^{(1)}}$
family} 
\label{sec:discrete_painlev_e_equations_in_the_lowercase_d_dpain_a__3_1_d__5_1_family}

To make this paper self-contained, we collect in this Appendix some of the basic facts about the geometry of the $D_{5}^{(1)}$-family
of Sakai surfaces and some standard discrete Painlev\'e equations associated with this surface family. The computations here
are standard (see \cite{KajNouYam:2017:GAOPE}, \cite{DzhTak:2018:OSAOSGTODPE}, \cite{DzhFilSto:2019:RCDOPHWDPE}) and are mostly omitted. 
We use $(q,p)$-coordinates for the standard example and
follow the standard reference \cite{KajNouYam:2017:GAOPE} for the choice of the standard point configuration and the root bases.

\subsection{The Point Configuration} 
\label{sub:the_point_configuration}
We start with the root basis of the surface sub-lattice that is given by the classes $\delta_{i}$ of the irreducible 
components of the anti-canonical divisor 
\begin{equation*}
	\delta = - \mathcal{K}_{\mathcal{X}} = 2 \mathcal{H}_{f} + 2 \mathcal{H}_{g} - \mathcal{E}_{1} 
	 - \mathcal{E}_{2} - \mathcal{E}_{3} - \mathcal{E}_{4} - \mathcal{E}_{5} - \mathcal{E}_{6} - \mathcal{E}_{7} - \mathcal{E}_{8}
	 = \delta_{0} + \delta_{1} + 2 \delta_{2} + 2 \delta_{3} + \delta_{4} + \delta_{5}.
\end{equation*}
The intersection configuration of those roots is given by the Dynkin diagram of type $D_{5}^{(1)}$, as shown on Figure~\ref{fig:d-roots-d51}.
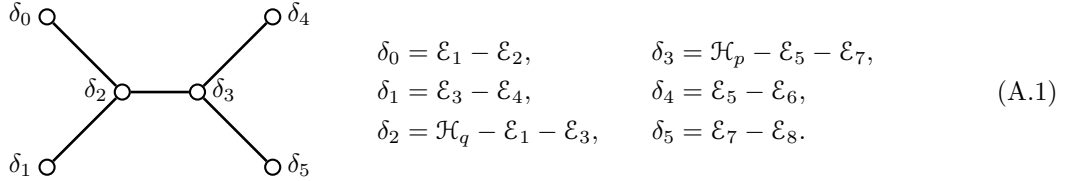
\begin{figure}[ht]
\begin{equation}\label{eq:d-roots-d51}			
	\raisebox{-32.1pt}{\begin{tikzpicture}[
			elt/.style={circle,draw=black!100,thick, inner sep=0pt,minimum size=2mm}]
		\path 	(-1,1) 	node 	(d0) [elt, label={[xshift=-10pt, yshift = -10 pt] $\delta_{0}$} ] {}
		        (-1,-1) node 	(d1) [elt, label={[xshift=-10pt, yshift = -10 pt] $\delta_{1}$} ] {}
		        ( 0,0) 	node  	(d2) [elt, label={[xshift=-10pt, yshift = -10 pt] $\delta_{2}$} ] {}
		        ( 1,0) 	node  	(d3) [elt, label={[xshift=10pt, yshift = -10 pt] $\delta_{3}$} ] {}
		        ( 2,1) 	node  	(d4) [elt, label={[xshift=10pt, yshift = -10 pt] $\delta_{4}$} ] {}
		        ( 2,-1) node 	(d5) [elt, label={[xshift=10pt, yshift = -10 pt] $\delta_{5}$} ] {};
		\draw [black,line width=1pt ] (d0) -- (d2) -- (d1)  (d2) -- (d3) (d4) -- (d3) -- (d5);
	\end{tikzpicture}} \qquad
			\begin{alignedat}{2}
			\delta_{0} &= \mathcal{E}_{1} - \mathcal{E}_{2}, &\qquad  \delta_{3} &= \mathcal{H}_{p} - \mathcal{E}_{5} - \mathcal{E}_{7},\\
			\delta_{1} &= \mathcal{E}_{3} - \mathcal{E}_{4}, &\qquad  \delta_{4} &= \mathcal{E}_{5} - \mathcal{E}_{6},\\
			\delta_{2} &= \mathcal{H}_{q} - \mathcal{E}_{1} - \mathcal{E}_{3}, &\qquad  \delta_{5} &= \mathcal{E}_{7} - \mathcal{E}_{8}.
			\end{alignedat}
\end{equation}
	\caption{The Surface Root Basis for the standard d-$\dPain{D_{5}^{(1)}}$ point configuration}
	\label{fig:d-roots-d51}	
\end{figure}

Using the action of the $\mathbf{PGL}_{2}(\mathbb{C})\times \mathbf{PGL}_{2}(\mathbb{C})$ gauge group 
we can put divisors $d_{2}$ and $d_{3}$, with $\delta_{i} = [d_{i}]$, to be
\begin{equation*}
	d_{2} = V(Q) = \{q = \infty\},\qquad d_{3} = V(P) = \{p = \infty\}.
\end{equation*} 
This reduces the gauge group action to that of a four-parameter subgroup, $(q,p)\mapsto (\lambda q + \mu, \zeta p + \xi)$. 
The corresponding point configuration and the Sakai surface are shown on Figure~\ref{fig:surface-d5}.
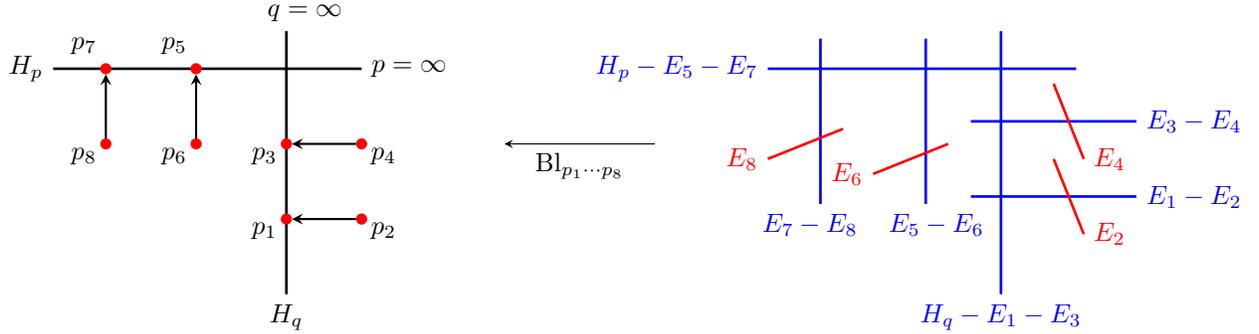
\begin{figure}[ht]
	\begin{tikzpicture}[>=stealth,basept/.style={circle, draw=red!100, fill=red!100, thick, inner sep=0pt,minimum size=1.2mm}]
		\begin{scope}[xshift = -1cm]
			\draw [black, line width = 1pt] 	(3.6,2.5) 	-- (-0.5,2.5)	node [left] {$H_{p}$} node[pos=0, right] {$p=\infty$};
			\draw [black, line width = 1pt] 	(2.6,3) -- (2.6,-0.5)		node [below] {$H_{q}$} node[pos=0, above, xshift=7pt] {$q=\infty$};
		
			\node (p1) at (2.6,0.5) [basept,label={[below left] $p_{1}$}] {};
			\node (p2) at (3.6,0.5) [basept,label={[below right] $p_{2}$}] {};
			\node (p3) at (2.6,1.5) [basept,label={[below left] $p_{3}$}] {};
			\node (p4) at (3.6,1.5) [basept,label={[below right] $p_{4}$}] {};
			\node (p5) at (1.4,2.5) [basept,label={[above left] $p_{5}$}] {};
			\node (p6) at (1.4,1.5) [basept,label={[below left] $p_{6}$}] {};
			\node (p7) at (0.2,2.5) [basept,label={[above left] $p_{7}$}] {};
			\node (p8) at (0.2,1.5) [basept,label={[below left] $p_{8}$}] {};
			\draw [line width = 0.8pt, ->] (p2) -- (p1);
			\draw [line width = 0.8pt, ->] (p4) -- (p3);
			\draw [line width = 0.8pt, ->] (p6) -- (p5);
			\draw [line width = 0.8pt, ->] (p8) -- (p7);
		\end{scope}
	
		\draw [->] (6.5,1.5)--(4.5,1.5) node[pos=0.5, below] {$\operatorname{Bl}_{p_{1}\cdots p_{8}}$};
	
		\begin{scope}[xshift = 8.5cm]
			\draw [blue, line width = 1pt] 	(3.6,2.5) 	-- (-0.5,2.5)	node [left] {$H_{p}-E_{5} - E_{7}$} node[pos=0, right] {};
			\draw [blue, line width = 1pt] 	(2.6,3) -- (2.6,-0.5)			node [below] {$H_{q}-E_{1}-E_{3}$} node[pos=0, above, xshift=7pt] {};

			\draw [blue,line width = 1pt] (2.2,0.8) -- (4.4,0.8) node [right] {$E_{1} - E_{2}$};
			\draw [red,line width = 1pt] (3.3,1.3) -- (3.7,0.3) node [right] {$E_{2}$};
			\draw [blue,line width = 1pt] (2.2,1.8) -- (4.4,1.8) node [right] {$E_{3} - E_{4}$};
			\draw [red,line width = 1pt] (3.3,2.3) -- (3.7,1.3) node [right] {$E_{4}$};
			\draw [blue,line width = 1pt] (1.6,0.7) -- (1.6,2.9) node [pos=0,below] {$\phantom{E}E_{5} - E_{6}$};
			\draw [red,line width = 1pt] (0.9,1.1) -- (1.9,1.5) node [pos=0,left] {$E_{6}$};
			\draw [blue,line width = 1pt] (0.2,0.7) -- (0.2,2.9) node [pos=0,below] {$E_{7} - E_{8}\phantom{E}$};
			\draw [red,line width = 1pt] (-0.5,1.3) -- (0.5,1.7) node [pos=0,left] {$E_{8}$};
		\end{scope}
	\end{tikzpicture}
	\caption{The model Sakai Surface for the d-$P\left(A_{3}^{(1)}/D_{5}^{(1)}\right)$ example}
	\label{fig:surface-d5}
\end{figure}

This point configuration can be parameterized by eight parameters $b_{1},\ldots, b_{8}$ as follows:
\begin{alignat*}{4}
	&p_{1}(\infty,b_{1})&&\leftarrow p_{2}(\infty,b_{1};q(p-b_{1})=b_{2}),	
	&\quad &p_{5}(b_{5},\infty)	&&\leftarrow p_{6}(b_{5},\infty;(q-b_{5})p=b_{6}),\\
	&p_{3}(\infty,b_{3})&&\leftarrow p_{4}(\infty,b_{3};q(p-b_{3})=b_{4}),	
	&\quad &p_{7}(b_{7},\infty)	&&\leftarrow p_{8}(b_{7},\infty;(q-b_{7})p=b_{8}).
\end{alignat*}
The four-parameter gauge group above acts on these configurations via
	\begin{equation}\label{eq:gauge-d5}
			\left(\begin{matrix}
				b_{1} & b_{2} & b_{3} & b_{4}\\
				b_{5} & b_{6} & b_{7} & b_{8}
			\end{matrix};  \begin{matrix}
				q \\ p
			\end{matrix}\right) \sim \left(\begin{matrix}
				\zeta b_{1} + \xi & \lambda \zeta b_{2} & \zeta b_{3}  + \xi & \lambda \zeta b_{4} \\
				\lambda b_{5} + \mu & \lambda \zeta b_{6}  & \lambda b_{7} + \mu& \lambda \zeta b_{8}
			\end{matrix}; \begin{matrix}
				\lambda q + \mu\\ \zeta g +\xi
			\end{matrix}\right),\,\lambda,\zeta\neq0,
	\end{equation}
and so the true number of parameters is four. The correct gauge-invariant parameterization is given by 
the \emph{root variables} that we now describe.

\subsection{The Period Map and the Root Variables} 
\label{sub:the_period_map_and_the_root_variables}
To define the root variables we begin by choosing a root basis in the 
\emph{symmetry sub-lattice} $Q = \Pi(R^{\perp}) \triangleleft \operatorname{Pic}(\mathcal{X})$ 
and defining the symplectic form $\omega$
whose polar divisor $-K_{\mathcal{X}}$ is the configuration of $-2$-curves shown on Figure~\ref{fig:surface-d5}. 
For the symmetry root basis we take the same basis as in \cite{KajNouYam:2017:GAOPE}, see Figure~\ref{fig:a-roots-a3}.

\begin{figure}[ht]
\begin{equation}\label{eq:a-roots-a3}			
	\raisebox{-32.1pt}{\begin{tikzpicture}[
			elt/.style={circle,draw=black!100,thick, inner sep=0pt,minimum size=2mm}]
		\path 	(-1,1) 	node 	(a0) [elt, label={[xshift=-10pt, yshift = -10 pt] $\alpha_{0}$} ] {}
		        (-1,-1) node 	(a1) [elt, label={[xshift=-10pt, yshift = -10 pt] $\alpha_{1}$} ] {}
		        ( 1,-1) node  	(a2) [elt, label={[xshift=10pt, yshift = -10 pt] $\alpha_{2}$} ] {}
		        ( 1,1) 	node 	(a3) [elt, label={[xshift=10pt, yshift = -10 pt] $\alpha_{3}$} ] {};
		\draw [black,line width=1pt ] (a0) -- (a1) -- (a2) --  (a3) -- (a0); 
	\end{tikzpicture}} \qquad
			\begin{alignedat}{2}
			\alpha_{0} &= \mathcal{H}_{p} - \mathcal{E}_{1} - \mathcal{E}_{2}, &\qquad  \alpha_{2} &= \mathcal{H}_{p} - \mathcal{E}_{3} - \mathcal{E}_{4},\\
			\alpha_{1} &= \mathcal{H}_{q} - \mathcal{E}_{5} - \mathcal{E}_{6}, &\qquad  \alpha_{3} &= \mathcal{H}_{q} - \mathcal{E}_{7} - \mathcal{E}_{8}.
			\\[5pt]
			\delta & = \mathrlap{\alpha_{0} + \alpha_{1} +  \alpha_{2} + \alpha_{3}.} 
			\end{alignedat}
\end{equation}
	\caption{The Standard Root Basis for the d-$P\left(A_{3}^{(1)}\right)$ Symmetry Sub-lattice}
	\label{fig:a-roots-a3}	
\end{figure}

A symplectic form $\omega\in -\mathcal{K}_{\mathcal{X}}$ such that 
$[\omega] = \delta_{0} + \delta_{1} + 2 \delta_{2} + 2\delta_{3} + \delta_{4} + \delta_{5}$ can be given in local 
coordinate charts as
\begin{equation}\label{eq:symp-form}
	\omega = k dq\wedge dp = - k \frac{dQ\wedge dp}{Q^{2}} = - k \frac{dq\wedge dP}{P^{2}} = k \frac{dQ \wedge dP}{Q^{2} P^{2}} 
	= - k \frac{du_{i}\wedge dv_{i}}{u_{i}} = - k \frac{dU_{j}\wedge d V_{j}}{V_{j}},
\end{equation}
where, as usual, $Q = 1/q$, $P = 1/p$ are the coordinates centered at infinity, the blowup coordinates $u_{i}$, $v_{i}$ at the 
points $p_{i}$, $i=1,3$, are given by $Q = u_{i}$, $p = b_{i} + u_{i} v_{i}$, and the blowup coordinates
$(U_{j}, V_{j})$ at the points $p_{j}$, $j=5,7$,
are given by $q = b_{j} + U_{j}V_{j}$ and $P = V_{j}$;  $k$ is some non-zero proportionality constant that we normalize later. 
Then we have the following Lemma.

\begin{lemma}\label{lem-period_map-d4} 
	\qquad
	\begin{enumerate}[(i)]
		\item The residue of the symplectic form $\omega$ along the irreducible components of the polar divisor is given by
		\begin{equation}
			\operatorname{res}_{d_{0}} \omega = -k\, dv_{1},\   
			\operatorname{res}_{d_{1}} \omega = - k\, dv_{3},\   
			\operatorname{res}_{d_{2}} \omega = \operatorname{res}_{d_{3}} \omega = 0,\  
			\operatorname{res}_{d_{4}} \omega = k\, dU_{5},\  
			\operatorname{res}_{d_{5}} \omega =  k\, dU_{7}.
		\end{equation}
		\item The root variables $a_{i}$ are given by 
		\begin{equation}\label{eq:d4-root_vars}
			a_{0} = k b_{2},\quad a_{1} = - k b_{6},\quad a_{2} = k b_{4},\quad a_{3} = - k b_{8}.
		\end{equation}
		It is convenient to take $k=-1$. We can then use the gauge action \eqref{eq:gauge-d5} to normalize $b_{3} = b_{5} = 0$, 
		$b_{7} = 1$, and 
		$\chi(\delta) = a_{0} + a_{1} + a_{2} + a_{3} = 1$. In view of the relation of this example to differential Painlev\'e equations,
		it is also convenient to denote $b_{1}$ by $-t$. Then we get the following parameterization of this point configuration 
		in terms of root variables:
		\begin{equation}\label{eq:d4-root_var_par}
			b_{1} = -t,\quad b_{2} = -a_{0},\quad b_{3} = 0,\quad b_{4} = - a_{2},
			\quad b_{5} = 0,\quad b_{6} = a_{1},\quad b_{7} = 1,\quad b_{8} = a_{3}.
		\end{equation}  
		Note that if we use the notation 
		\begin{equation*}
			p_{12}\left(\frac{1}{\varepsilon},-t- \varepsilon a_{0}\right),\quad p_{34}\left(\frac{1}{\varepsilon},-\varepsilon a_{2}\right),\quad
			p_{56}\left(a_{1} \varepsilon, \frac{1}{\varepsilon}\right),\quad p_{78}\left(1 + a_{3} \varepsilon,\frac{1}{\varepsilon}\right),
		\end{equation*}
		and impose the normalization $a_{0} + a_{1} + a_{2} + a_{3} = 1$, 
		we get exactly the parameterization of the point configuration in section 8.2.18 of \cite{KajNouYam:2017:GAOPE}.
	\end{enumerate}
\end{lemma}

\subsection{The Extended Affine Weyl Symmetry Group} 
\label{sub:the_extended_affine_weyl_symmetry_group}

We now describe the birational representation of the extended
affine Weyl symmetry group $\widetilde{W}\left(A_{3}^{(1)}\right) = \operatorname{Aut}\left(A_{3}^{(1)}\right) \ltimes W\left(A_{3}^{(1)}\right)$,
which is a \emph{semi-direct product} of the usual affine Weyl group $W\left(A_{3}^{(1)}\right)$ and the group of Dynkin diagram automorphisms
$\operatorname{Aut}\left(A_{3}^{(1)}\right)\simeq \mathbb{D}_{4}$.

The abstract affine Weyl group $W\left(A_{3}^{(1)}\right)$ is defined in terms of generators $w_{i} = w_{\alpha_{i}}$ and relations that 
are encoded by the affine Dynkin diagram $A_{3}^{(1)}$,
\begin{equation*}
	W\left(A_{3}^{(1)}\right) = W\left(\raisebox{-20pt}{\begin{tikzpicture}[
			elt/.style={circle,draw=black!100,thick, inner sep=0pt,minimum size=1.5mm}]
		\path 	(-0.5,0.5) 	node 	(a0) [elt, label={[xshift=-10pt, yshift = -10 pt] $\alpha_{0}$} ] {}
		        (-0.5,-0.5) node 	(a1) [elt, label={[xshift=-10pt, yshift = -10 pt] $\alpha_{1}$} ] {}
		        (0.5,-0.5) 	node 	(a2) [elt, label={[xshift=10pt, yshift = -10 pt] $\alpha_{2}$} ] {}
		        (0.5,0.5) 	node 	(a3) [elt, label={[xshift=10pt, yshift = -10 pt] $\alpha_{3}$} ] {};
		\draw [black,line width=1pt ] (a0) -- (a1) -- (a2)--(a3) -- (a0); 
	\end{tikzpicture}} \right)
	=
	\left\langle w_{0},\dots, w_{4}\ \left|\ 
	\begin{alignedat}{2}
    w_{i}^{2} = e,\quad  w_{i}\circ w_{j} &= w_{j}\circ w_{i}& &\text{ when 
   				\raisebox{-0.08in}{\begin{tikzpicture}[
   							elt/.style={circle,draw=black!100,thick, inner sep=0pt,minimum size=1.5mm}]
   						\path   ( 0,0) 	node  	(ai) [elt] {}
   						        ( 0.5,0) 	node  	(aj) [elt] {};
   						\draw [black] (ai)  (aj);
   							\node at ($(ai.south) + (0,-0.2)$) 	{$\alpha_{i}$};
   							\node at ($(aj.south) + (0,-0.2)$)  {$\alpha_{j}$};
   							\end{tikzpicture}}}\\
    w_{i}\circ w_{j}\circ w_{i} &= w_{j}\circ w_{i}\circ w_{j}& &\text{ when 
   				\raisebox{-0.17in}{\begin{tikzpicture}[
   							elt/.style={circle,draw=black!100,thick, inner sep=0pt,minimum size=1.5mm}]
   						\path   ( 0,0) 	node  	(ai) [elt] {}
   						        ( 0.5,0) 	node  	(aj) [elt] {};
   						\draw [black] (ai) -- (aj);
   							\node at ($(ai.south) + (0,-0.2)$) 	{$\alpha_{i}$};
   							\node at ($(aj.south) + (0,-0.2)$)  {$\alpha_{j}$};
   							\end{tikzpicture}}}
	\end{alignedat}\right.\right\rangle. 
\end{equation*} 
The natural action of this group on $\operatorname{Pic}(\mathcal{X})$ is given by reflections in the 
roots $\alpha_{i}$, 
\begin{equation}\label{eq:root-refl}
	w_{i}(\mathcal{C}) = w_{\alpha_{i}}(\mathcal{C}) = \mathcal{C} - 2 
	\frac{\mathcal{C}\bullet \alpha_{i}}{\alpha_{i}\bullet \alpha_{i}}\alpha_{i}
	= \mathcal{C} + \left(\mathcal{C}\bullet \alpha_{i}\right) \alpha_{i},\qquad \mathcal{C}\in \operatorname{Pic(\mathcal{X})},
\end{equation}
which can be extended to an action on point configurations by elementary birational maps (which lifts to 
isomorphisms $w_{i}: \mathcal{X}_{\mathbf{b}}\to \mathcal{X}_{\overline{\mathbf{b}}}$ on the family of Sakai's surfaces),
this is known as a birational representation of $W\left(A_{3}^{(1)}\right)$.

\begin{theorem}\label{thm:bir-weyl-d5}
	Reflections $w_{i}$ on $\operatorname{Pic}(\mathcal{X})$ are induced by the elementary 
	birational mappings given below, and also denoted by $w_{i}$, on the family $\mathcal{X}_{\mathbf{b}}$. To ensure the group structure, 
	we require that each mapping preserves our normalization, and so it is enough to describe the mappings in terms of the root variables
	(note that the parameter $t$ can also change when we consider the Dynkin diagram automorphisms, so it is convenient to include it 
	among the root variables):
	\begin{alignat}{2}
		w_{0}&: 
		\left(\begin{matrix} a_{0} & a_{1} \\ a_{2} & a_{3} \end{matrix}\ ;\ t\ ;
		\begin{matrix} q \\ p \end{matrix}\right) 
		&&\mapsto 
		\left(\begin{matrix} -a_{0} & a_{0} + a_{1} \\ a_{2} & a_{0} + a_{3} \end{matrix}\ ;\ t\ ;
		\begin{matrix} \displaystyle q + \frac{a_{0}}{p + t} \\ p \end{matrix}\right), \\
		w_{1}&: \left(\begin{matrix} a_{0} & a_{1} \\ a_{2} & a_{3} \end{matrix}\ ;\ t\ ;
		\begin{matrix} q \\ p \end{matrix}\right)
		&&\mapsto 
		\left(\begin{matrix} a_{0} + a_{1} & -a_{1} \\ a_{1} + a_{2} & a_{3} \end{matrix}\ ;\ t\ ;
		\begin{matrix}  q \\ \displaystyle p - \frac{a_{1}}{q} \end{matrix}\right), \\
		w_{2}&: 
		\left(\begin{matrix} a_{0} & a_{1} \\ a_{2} & a_{3} \end{matrix}\ ;\ t\ ;
		\begin{matrix} q \\ p \end{matrix}\right) 
		&&\mapsto 
		\left(\begin{matrix} a_{0} & a_{1} + a_{2} \\ -a_{2} & a_{2} + a_{3} \end{matrix}\ ;\ t\ ;
		\begin{matrix} \displaystyle q + \frac{a_{2}}{p}\\ p \end{matrix}\right), \\
		w_{3}&: 
		\left(\begin{matrix} a_{0} & a_{1} \\ a_{2} & a_{3} \end{matrix}\ ;\ t\ ;
		\begin{matrix} q \\ p \end{matrix}\right) 
		&&\mapsto 
		\left(\begin{matrix} a_{0}+a_{3} & a_{1} \\ a_{2}+a_{3} & -a_{3} \end{matrix}\ ;\ t\ ;
		\begin{matrix} q \\ \displaystyle  p - \frac{a_{3}}{q-1} \end{matrix}\right).
	\end{alignat}	
\end{theorem}

It is clear that the group of Dynkin diagram automorphisms $\operatorname{Aut}\left(A_{3}^{(1)}\right)\simeq \mathbb{D}_{4}$, 
so we only describe two generators $\sigma_{1}$, $\sigma_{2}$, as well as one more automorphism $\sigma_{3}$
that we need.

\begin{theorem}\label{thm:bir-aut-d4}
	Consider the automorphisms $\sigma_{1},\ldots ,\sigma_{3}$ of $\operatorname{Aut}\left(A_{3}^{(1)}\right)$ 
	that act on the symmetry and the surface root bases as follows (here we use the standard cycle 
	notations for permutations):
	\begin{equation}
		\sigma_{1} = (\alpha_{0}\alpha_{3})(\alpha_{1}\alpha_{2})=
		(\delta_{0}\delta_{5})(\delta_{1}\delta_{4})(\delta_{2}\delta_{3}),\qquad 
		\sigma_{2} = (\alpha_{0}\alpha_{2})=(\delta_{0}\delta_{1}), \qquad
		\sigma_{3} = (\alpha_{1}\alpha_{3})=(\delta_{4}\delta_{5}).
	\end{equation}
	Then $\sigma_{i}$ act on the Picard lattice as
	\begin{equation*}
			\sigma_{1} = (\mathcal{E}_{1}\mathcal{E}_{7})(\mathcal{E}_{2}\mathcal{E}_{8})(\mathcal{E}_{3}\mathcal{E}_{5})
			(\mathcal{E}_{4}\mathcal{E}_{6})w_{\rho},\qquad 
			\sigma_{2} = (\mathcal{E}_{1}\mathcal{E}_{3})(\mathcal{E}_{2}\mathcal{E}_{4}), \qquad
			\sigma_{3} = (\mathcal{E}_{5}\mathcal{E}_{7})(\mathcal{E}_{6}\mathcal{E}_{8}),			
	\end{equation*}
	where $w_{\rho}$ is a reflection  \eqref{eq:root-refl} in the root $\rho = \mathcal{H}_{q} - \mathcal{H}_{p}$ 
	(note also that a transposition 
	$(\mathcal{E}_{i} \mathcal{E}_{j})$ is induced by a reflection in the root $\mathcal{E}_{i} - \mathcal{E}_{j}$).
	The induced elementary 	birational mappings are then given by the following expressions:
	\begin{alignat}{2}
		\sigma_{1}&: 
		\left(\begin{matrix} a_{0} & a_{1} \\ a_{2} & a_{3} \end{matrix}\ ;\ t\ ;
		\begin{matrix} q \\ p \end{matrix}\right) 
		&&\mapsto 
		\left(\begin{matrix} a_{3} & a_{2} \\ a_{1} & a_{0}  \end{matrix}\ ;\ -t\ ;
		\begin{matrix} \displaystyle -\frac{p}{ t} \\ q t \end{matrix}\right), \\
		\sigma_{2}&: 
		\left(\begin{matrix} a_{0} & a_{1} \\ a_{2} & a_{3} \end{matrix}\ ;\ t\ ;
		\begin{matrix} q \\ p \end{matrix}\right) 
		&&\mapsto 
		\left(\begin{matrix} a_{2} & a_{1} \\ a_{0} & a_{3}  \end{matrix}\ ;\ -t\ ;
		\begin{matrix} q \\ p + t \end{matrix}\right), \\
		\sigma_{3}&: 
		\left(\begin{matrix} a_{0} & a_{1} \\ a_{2} & a_{3} \end{matrix}\ ;\ t\ ;
		\begin{matrix} q \\ p \end{matrix}\right) 
		&&\mapsto 
		\left(\begin{matrix} a_{0} & a_{3} \\ a_{2} & a_{1}  \end{matrix}\ ;\ -t\ ;
		\begin{matrix} 1-q \\ -p  \end{matrix}\right).
	\end{alignat}	
\end{theorem}

Finally, the semi-direct product structure is defined by the action of $\sigma \in \operatorname{Aut}\left(A_{3}^{(1)}\right)$ on 
$W\left(A_{3}^{(1)}\right)$ via $w_{\sigma(\alpha_{i})} = \sigma w_{\alpha_{i}} \sigma^{-1}$.

\subsection{Some standard discrete d-$\dPain{A_{3}^{(1)}/D_{5}^{(1)}}$ equations} 
\label{sub:some_standard_discrete_d_dpain_a__3_1_d__5_1_equations}
There are infinitely many different discrete Painlev\'e equations of the same type corresponding to the non-conjugate 
translations in the affine symmetry sub-lattice $Q$. Of those, we are interested in two particular equations that correspond to short translation 
vectors. One is equation~(8.23) in \cite[Section 8.1.17]{KajNouYam:2017:GAOPE}, the other is the so-called d-$\Pain{IV}$ equation in 
\cite{Sak:2001:RSAWARSGPE}, which also appears in a slightly different form (2.33--2.34) in \cite{Sak:2007:PDPEATLF}. We label these equations 
by $\mathbf{[\overline{1}1 \overline{1}1]}$ and $\mathbf{[\overline{1}001]}$ respectively, based on the induced action of the dynamics on 
the symmetry roots (see below), which is unambiguous. In the above references 
these equations are presented in a geometric way as mappings, similar to our approach. However, both classes of equations were obtained earlier by 
Basil~Grammaticos, Alfred~Ramani, and their collaborators using the singularity confinement approach; in their papers these equations are presented as 
recurrences with particular coefficient evolution. Equation $\mathbf{[\overline{1}1 \overline{1}1]}$ first appeared in 
\cite{GraNijPap:1994:LSDPE} (where it was shown that this equation actually has d-$\Pain{III}$ and not d-$\Pain{IV}$ as a continuous limit)  
and equation $\mathbf{[\overline{1}001]}$ first appeared in  \cite{GraOhtRam:1998:DTCOTQVE}; see also \cite{TokGraRam:2002:CDPE} where 
equations (3.1--3.2) is essentially the mapping \eqref{eq:dPD5-KNY} and
equations (3.24ab) is essentially the mapping \eqref{eq:dPD5-Sakai}\footnote{We thank A.~Ramani for his help with historical references.}. 

Note that equations $\mathbf{[\overline{1}1 \overline{1}1]}$ and $\mathbf{[\overline{1}001]}$ 
\emph{are not equivalent} --- this can be seen, for example, from the length of the corresponding words in the 
extended affine Weyl group, or from the lengths of the corresponding translations, or, probably in the simplest possible way, by computing the 
Jordan form of the matrix description of the evolution on $\operatorname{Pic}(\mathcal{X})$.

\subsubsection{The $[\mathbf{\overline{1}1 \overline{1}1}]$ discrete Painlev\'e equation on the $D_{5}^{(1)}$ surface} 
\label{ssub:kajiwara_noumi_yamada_or_d_pain_iii_equation_on_the_d__5_1_surface}
In \cite{KajNouYam:2017:GAOPE}, the standard example of a discrete Painlev\'e equation on the $D_{5}^{(1)}$-surface is given in Section~8.1.17
equation~(8.23), and it has the following \emph{additive} form, when written in coordinates $(q,p)$:
\begin{equation}\label{eq:dPD5-KNY}
	\overline{q} + q = 1 - \frac{a_{2}}{p} - \frac{a_{0}}{p+t},\qquad 
	p + \underline{p} = -t + \frac{a_{1}}{q} + \frac{a_{3}}{q-1}
\end{equation}
with the root variable evolution and normalization given by 
\begin{equation}\label{eq:dPD5-rv-evol}
	\overline{a}_{0} = a_{0} + 1, \quad \overline{a}_{1} = a_{1}-1, \quad \overline{a}_{2} = a_{2} + 1,\quad \overline{a}_{3} = a_{3} - 1,\qquad
	 a_{0} + a_{1} + a_{2}  + a_{3} = 1.
\end{equation}
For this equation, the geometry of the corresponding point configuration is shown on Figure~\ref{fig:surface-d5}, with the parameterization 
by the root variables is given by \eqref{eq:d4-root_var_par}. 
From the root variable evolution \eqref{eq:dPD5-rv-evol} we immediately see that the corresponding translation on the root lattice is 
\begin{equation}\label{eq:dPD5-transl-KNY}
	\varphi_{*}: \upalpha =  \langle \alpha_{0}, \alpha_{1}, \alpha_{2}, \alpha_{3} \rangle
	\mapsto \varphi_{*}(\upalpha) = \upalpha + \langle -1,1,-1,1 \rangle \delta,
\end{equation}
which explains our labeling for this equation (we use $\overline{1}$ instead of $-1$ for compactness).
Using the standard techniques, see \cite{DzhTak:2018:OSAOSGTODPE} for a detailed example, we get the following decomposition of $\varphi$ in 
terms of the generators of $\widetilde{W}\left(A_{3}^{(1)}\right)$:
\begin{equation}\label{eq:dPD5-decomp-KNY}
	\varphi = \sigma_{3}\sigma_{2} w_{3} w_{1} w_{2} w_{0}.
\end{equation}
Note that equations \eqref{eq:dPD5-KNY} naturally define two \emph{half-maps}, $\varphi_{1}: (q,p)\to(\overline{q},-p)$ and 
$\varphi_{2}:(q,p)\to(q,-\underline{p})$ (the additional negative sign here is related to the M\"obius group gauge action as explained in 
\cite[Section~2.9]{DzhFilSto:2019:RCDOPHWDPE}), and the mapping $\varphi$ that we are interested in is 
$\varphi = (\overline{\varphi_{2}})^{-1}\circ \varphi_{1}$. These individual mappings decompose as $\varphi_{1} = \sigma_{3}w_{2}w_{0}$
and $\varphi_{3} = \sigma_{2}w_{3}w_{1}$.

\subsubsection{The $[\mathbf{\overline{1}001}]$ discrete Painlev\'e equation  on the $D_{5}^{(1)}$ surface} 
\label{ssub:sakai_or_d_pain_iv_equation_on_the_d__5_1_surface}
\begin{figure}[ht]
	\begin{tikzpicture}[>=stealth,basept/.style={circle, draw=red!100, fill=red!100, thick, inner sep=0pt,minimum size=1.2mm}]
		\begin{scope}[xshift = 0cm]
			\draw [black, line width = 1pt]  	(4,0) 	-- (-0.5,0) 	node [left] {$H_{g}$} node[pos=0, right] {$g=0$};
			\draw [black, line width = 1pt] 	(4,2.5) -- (-0.5,2.5)	node [left] {$H_{g}$} node[pos=0, right] {$g=\infty$};
			\draw [black, line width = 1pt] 	(0,3) -- (0,-0.5)		node [below] {$H_{f}$} node[pos=0, above, xshift=-7pt] {$f=0$};
			\draw [black, line width = 1pt] 	(3,3) -- (3,-0.5)		node [below] {$H_{f}$} node[pos=0, above, xshift=7pt]  {$f=\infty$};
			
			\node (q1) at (0,2.5) 		[basept,label={[left, yshift = -8pt] 	$\pi_{1}$}] {};
			\node (q2) at (0.7,3) 		[basept,label={[above]	 		$\pi_{2}$}] {};
			\node (q3) at (1.3,3)	 	[basept,label={[above]	 		$\pi_{3}$}] {};
			\node (q4) at (2,2.5) 		[basept,label={[yshift=-18pt] 			$\pi_{4}$}] {};
			\node (q5) at (2,3) 		[basept,label={[above]	 		$\pi_{5}$}] {};
			\node (q6) at (3,0.8)		[basept,label={[right] $\pi_{6}$}] {};
			\node (q7) at (3,1.6)		[basept,label={[right] $\pi_{7}$}] {};
			\node (q8) at (0,0) 		[basept,label={[left, yshift = -8pt] 	$\pi_{8}$}] {};
			\draw [line width = 0.8pt, ->] (q3) edge (q2) (q2) edge (q1);
			\draw [line width = 0.8pt, ->] (q5) edge (q4);
		\end{scope}
	
		\draw [->] (6.3,1.5)--(4.3,1.5) node[pos=0.5, below] {$\operatorname{Bl}_{\pi_{1}\cdots \pi_{8}}$};
	
		\begin{scope}[xshift = 7.5cm]
			\draw [blue, line width = 1pt] 	(0,3) -- (0,-0.5)		node [below] {$H_{f} - K_{1} - K_{8}$};
			\draw [blue, line width = 1pt] 	(4,3) -- (4,-0.5)		node [below] {$H_{f} - K_{6} - K_{7}$};
			\draw [red, line width = 1 pt] (3.8,0.6) -- (4.6,1) node [right] {$K_{6}$};
			\draw [red, line width = 1 pt] (3.8,1.4) -- (4.6,1.8) node [right] {$K_{7}$};
			
			\draw [red, line width = 1pt ] (1.5,-1) .. controls (1.7,-0.5) and (1.9,0) .. (4.5,0) node [right] {$H_{g} - K_{8}$};
			\draw [red, line width = 1pt ] (-0.5,1)  -- (2,-1) node [right] {$K_{8}$};
			\draw [blue, line width = 1pt ] (-0.5,1.5)  -- (2,3.5) node [right] {$K_{1} - K_{2}$};
			\draw [blue, line width = 1pt ] (0.5,3)  -- (1.5,1) node [below, xshift=-8pt] {$K_{2} - K_{3}$};
			\draw [blue, line width = 1pt ] (3.5,3)  -- (2.5,1) node [below, xshift=8pt] {$K_{4} - K_{5}$};
			\draw [red, line width = 1pt ] (0.7,1.2)  -- (1.5,2) node [above] {$K_{3}$};
			\draw [red, line width = 1pt ] (3.3,1.2)  -- (2.5,2) node [above] {$K_{5}$};
			\draw [blue, line width = 1pt ] (1.5,3.5) .. controls (1.7,3) and (1.9,2.5) .. (4.5,2.5) node [right] {$H_{g} - K_{1} - K_{4}$};
			
		\end{scope}
	\end{tikzpicture}
	\caption{The Sakai Surface for the d-$\Pain{IV}$ example}
	\label{fig:surface-d5-Sakai}
\end{figure}
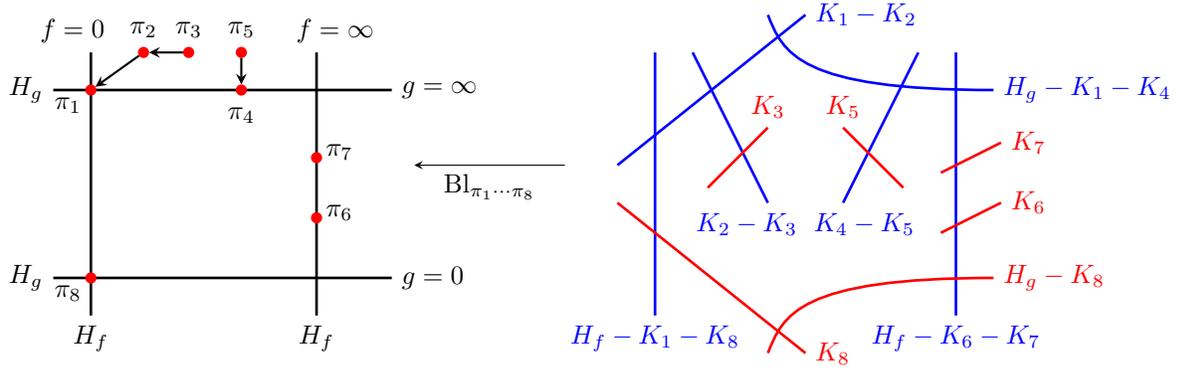

In \cite{Sak:2001:RSAWARSGPE}, the following mapping $\phi:(f,g)\to(\overline{f},\overline{g})$, 
written in the \emph{multiplicative-additive} form, is called a d-$\Pain{IV}$ equation on the $D_{5}^{(1)}$ surface:
\begin{equation}\label{eq:dPD5-Sakai}
	\overline{f} f = \frac{s \overline{g}}{(\overline{g} - a_{3} + \lambda)(\overline{g} + a_{0} + \lambda)},\qquad 
	\overline{g} + g = \frac{s}{f} + \frac{a_{1}+ a_{0}}{1 - f} - \lambda + a_{3} - a_{0},
\end{equation}
where $\lambda = a_{0} + a_{1} + a_{2} + a_{3}$ (without loss of generality it can be normalized to $\lambda = 1$), 
and the root variable evolution is given by $\overline{a}_{0} = a_{0} + \lambda$ and $\overline{a}_{3} = a_{3} - \lambda$. From the
root variable evolution we see that the corresponding translation on the root lattice is 
\begin{equation}\label{eq:dPD5-transl-Sakai}
	\phi_{*}: \upalpha =  \langle \alpha_{0}, \alpha_{1}, \alpha_{2}, \alpha_{3} \rangle
	\mapsto \varphi_{*}(\upalpha) = \upalpha + \langle -1,0,0,1 \rangle \delta.
\end{equation} This map can be written in terms of generators as 
\begin{equation}\label{eq:dPD5-decomp-Sakai}
	\phi = \sigma_{3}\sigma_{1}w_{2}w_{1}w_{0} = w_{3} w_{2} w_{1} \sigma_{3} \sigma_{1},
\end{equation}
which is the same as given in Sakai's paper.
However, the geometry of that example is slightly different from our reference model on Figure~\ref{fig:surface-d5} 
and is given on Figure~\ref{fig:surface-d5-Sakai}. This geometry can be 
matched to the standard one with the change of basis on $\operatorname{Pic}(\mathcal{X})$ given by 
\begin{align*}\label{eq:change-basis-Sakai-KNY}
	\mathcal{H}_{f} & = \mathcal{H}_{q} + \mathcal{H}_{p} - \mathcal{E}_{1} - \mathcal{E}_{2}, & \qquad 
	\mathcal{H}_{q} & = \mathcal{H}_{f} + \mathcal{H}_{g} - \mathcal{K}_{6} - \mathcal{K}_{8},\\
	\mathcal{H}_{g} & = \mathcal{H}_{q} + \mathcal{H}_{p}  - \mathcal{E}_{1} -\mathcal{E}_{7},  & \qquad 	
	\mathcal{H}_{p} &= \mathcal{H}_{f} + \mathcal{H}_{g} - \mathcal{K}_{1} - \mathcal{K}_{6},\\
	\mathcal{K}_{1} & = \mathcal{H}_{q} - \mathcal{E}_{1}, & \qquad 
	\mathcal{E}_{1} & = \mathcal{H}_{f} + \mathcal{H}_{g} - \mathcal{K}_{1} - \mathcal{K}_{6} - \mathcal{K}_{8} ,\\
	\mathcal{K}_{2} & = \mathcal{E}_{3}, & \qquad \mathcal{E}_{2} & = \mathcal{H}_{g} - \mathcal{K}_{6},\\
	\mathcal{K}_{3} & = \mathcal{E}_{4}, & \qquad \mathcal{E}_{3} & = \mathcal{K}_{2},\\
	\mathcal{K}_{4} & = \mathcal{E}_{5}, & \qquad \mathcal{E}_{4} & = \mathcal{K}_{3},\\
	\mathcal{K}_{5} & = \mathcal{E}_{6}, & \qquad \mathcal{E}_{5} & = \mathcal{K}_{4},\\
	\mathcal{K}_{6} & =  \mathcal{H}_{q} + \mathcal{H}_{p} - \mathcal{E}_{1} - \mathcal{E}_{2} - \mathcal{E}_{7}, & \qquad 
	\mathcal{E}_{6} & = \mathcal{K}_{5},\\
	\mathcal{K}_{7} & = \mathcal{E}_{8}, & \qquad \mathcal{E}_{7} & = \mathcal{H}_{f} - \mathcal{K}_{6},\\
	\mathcal{K}_{8} & = \mathcal{H}_{p} - \mathcal{E}_{1}, & \qquad \mathcal{E}_{8} & = \mathcal{K}_{7}.
\end{align*}
Note that this change of basis is chosen in such a way as to match the root variables between the two examples, however the parameters $s$ and $t$
differ by a sign, $s = -t$. The corresponding change of variables is given by 
\begin{equation*}
	\left\{\begin{aligned}
		f(q,p)& = -\frac{p+t}{(q-1)(p+t) + a_{0}},\\
		g(q,p)& = (q-1)(p+t),
	\end{aligned}\right. \qquad 
	\left\{\begin{aligned}
		q(f,g)&= 1 - \frac{g}{f(g + a_{0})},\\
		p(f,g)&= s - f(g + a_{0}).
	\end{aligned}\right. 
\end{equation*}




\section*{Acknowledgements} 
\label{sec:acknowledgements} 
Yang Chen and Hu Jie were supported by the Macau Science 
and Technology Development Fund under grant numbers FDCT 130/2014/A3, FDCT 023/2017/A1 and by the University of Macau 
under grant numbers MYRG 2014-00011-FST, MYRG 2014-00004-FST.
Part of this work was done when Anton Dzhamay visited Shanghai University and the University of Macau and he would like to thank 
both of these Universities for their support and hospitality.  
We also thank Alfred Ramani, Tomoyuki Takenawa, and Ralph Willox for helpful comments and discussions.


\bibliographystyle{amsalpha}

\providecommand{\bysame}{\leavevmode\hbox to3em{\hrulefill}\thinspace}
\providecommand{\MR}{\relax\ifhmode\unskip\space\fi MR }
\newcommand{\etalchar}[1]{$^{#1}$}
\providecommand{\href}[2]{#2}

\end{document}